\documentclass{sig-alternate}
\usepackage{graphicx}
\usepackage{amssymb}
\usepackage{mathrsfs}
\usepackage{epsfig}
\usepackage{color}
\usepackage{amsmath}
\usepackage{amssymb}
\usepackage{graphicx}
\usepackage{times,url}
\usepackage{subfigure}
\usepackage{graphics,color}
\usepackage{multicol,algorithm,algorithmic}
\newtheorem{theorem}{Theorem}

\newtheorem{corollary}{Corollary}
\newtheorem{definition}{Definition}

\newcommand{\ignore}[1]{}

\newcommand{\C}{{\cal C}}
\newcommand{\EXP}{{\sf exp}}
\newcommand{\secure}{{\em secure}}
\newcommand{\infected}{{\em infected}}

\begin{document}

\conferenceinfo{HotSoS}{'14 Raleigh, NC USA}

\title{Characterizing the Power of Moving Target Defense via Cyber Epidemic Dynamics}

\numberofauthors{1}

\author{
\alignauthor
Yujuan Han${^\dag}{^\star}$ ~~~ Wenlian Lu${^\dag}{^\ddag}$ ~~~ Shouhuai Xu$^\star$ \\
~ \\
\affaddr{$^\dag$ School of Mathematical Sciences, Fudan University} \\
\affaddr{$^\star$ Department of Computer Science, University of Texas at San Antonio} \\
\affaddr{$^\ddag$ Department of Computer Science, University of Warwick}
}

\maketitle

\begin{abstract}
Moving Target Defense (MTD) can enhance
the resilience of cyber systems against attacks. Although there have been many MTD techniques,
there is no systematic understanding and {\em quantitative} characterization of the power of MTD.
In this paper, we propose to use a cyber epidemic dynamics approach to characterize the power of MTD.
We define and investigate two complementary measures that are applicable when the defender aims to deploy MTD to achieve a certain security goal.
One measure emphasizes the {\em maximum} portion of time during which the system can afford to stay in
an undesired configuration (or posture),
without considering the cost of deploying MTD.
The other measure emphasizes the {\em minimum} cost of deploying MTD, while accommodating that
the system has to stay in an undesired configuration (or posture) for a given portion of time.
Our analytic studies lead to algorithms for optimally deploying MTD.
\end{abstract}

\category{D.4.6}{Security and Protection}{}

\terms{Security, Theory}

\keywords{Moving target defense, cyber epidemic dynamics, epidemic threshold, security models, cybersecurity dynamics}

\section{Introduction}

Moving Target Defense (MTD) is believed to be a ``game changer" for cyber defense.
Although there have been many studies on specific MTD techniques,
the power of MTD is often demonstrated via simulation.
Rigorously characterizing the power of MTD is an important problem
and is closely related to the well known hard problem of security metrics.
In this paper, we initiate the study of a novel approach for characterizing the power of MTD.

\subsection{Our Contributions}

We propose to use the cyber epidemic dynamics approach to characterize the power of {\em classes} of MTD techniques.
We define and investigate two novel and complementary security measures that are applicable when using MTD to achieve a certain defense goal.
The first measure is centered on the {\em maximum} portion of time (in the equilibrium) during which
the system can afford to stay in an undesired/insecure configuration (or posture), without considering the cost of deploying MTD.
The second measure is centered on the {\em minimum} cost when the system has to stay in an undesired/insecure configuration (or posture) for a predetermined portion of time.
Our analytic studies lead to algorithms for deploying MTD in such optimal fashions.
To our knowledge, this is the first systematic study on characterizing the power of classes of MTD techniques.

\subsection{The Science}

Rigorously characterizing the power of MTD (or any defense) would be a core problem in
the science of security. Indeed, the present study initiates a paradigm for measuring the power of MTD (or other kinds of defense techniques whose deployment can
make a global difference). The paradigm is centered on measuring the degree of undersired/insecure configurations that can be tolerated by
deploying advanced defense such as MTD. The specific criterion used in the present paper, namely
that the attacks are eventually wiped out in a certain sense,
can be substituted by other criteria. One possible candidate is the containment of malware infections to a certain tolerable level \cite{XuTAAS2010}
(e.g., by appropriately choosing threshold cryptosystems).

\subsection{Related Work}
\label{sec:related-work}

The present paper does not introduce any new MTD technique.
Rather, it studies how to systematically characterize
the power of MTD and optimally launch MTD.
Existing studies for a similar purpose are often based on simulation \cite{Antonatos:2005:DAH:1103626.1103633,IPHopping,BBN-DYNAT,Jafarian:2012:ORH:2342441.2342467}.
There is effort at analytically evaluating the power of some {\em specific} MTD techniques
from a {\em localized} view \cite{DBLP:series/ais/Manadhata13,ZhuGameSec13};
in contrast, we characterize the power of {\em classes} of MTD techniques from a {\em global} view.

Cyber epidemic dynamics was rooted in biological epidemic dynamics \cite{McKendrick1926,Kermack1927}.
The first cyber epidemic models \cite{KephartOkland91,KephartOkland93}
were limited by their {\em homogeneity} assumption that each computer/node has the same effect on the others.
Recently, models that are more appropriate for studying cyber security problems have been proposed
\cite{Pastor2002a,WangSRDS03,WangTISSEC08,TowsleyInfocom05,VanMieghemIEEEACMTON09,XuTAAS12,XuTDSC12,XuTAAS2010}.
As we will elaborate later, the basic idea underlying these models is to use a graph-theoretic abstraction to represent
the {\em attack-defense} structure, and use {\em parameters} to represent attack and defense capabilities.
Cyber epidemic dynamics is a special kind of {\em cybersecurity dynamics} \cite{XuHotSOS14-cybersecurity-dynamics-poster},

We will use the cyber epidemic dynamics model in \cite{XuTAAS2010} as the starting point of our study.
This model \cite{XuTAAS2010} describes {\em reactive adaptive defense}
(i.e., the defender aims to adjust its defense to control/contain the global security state).
We extend this model to accommodate MTD, a kind of {\em proactive defense},
and the resulting model is analyzed using different skills \cite{Liberzon,Mariton}.
We mention that the effect of dynamic structures in cyber epidemic models is studied in \cite{FaloutsosECML-PKDD10},
where the structure dynamics however follows a {\em deterministic and periodic} process, rather than {\em adaptively} scheduled by using (for example) MTD.
We also mention that the effect of dynamic semi-heterogeneous structures (i.e., clustered networks), rather than arbitrary heterogeneous structures, is studied in \cite{Rami2013}.
These studies \cite{FaloutsosECML-PKDD10,Rami2013} consider static parameters only and do not have any of the measures we propose to use.

The rest of the paper is organized as follows.
In Section \ref{sec:models-accommodating-MTD}, we present a classification of MTD techniques and describe a cyber epidemic dynamics model that can accommodate MTD.
In Section \ref{sdvsss}, we characterize the power of MTD that induces dynamic parameters.
In Section \ref{dsvsss}, we characterize the power of MTD that induces dynamic attack-defense structures.
We discuss the limitations of the present study in Section \ref{sec:extension-and-discussion}.
We conclude the paper in Section \ref{sec:conclusion}.

\section{Cyber Epidemic Dynamics Model Accommodating Moving Target Defense}
\label{sec:models-accommodating-MTD}

\subsection{Three Classes of Moving Target Defense Techniques}

As mentioned in Section \ref{sec:related-work} and elaborated later, cyber epidemic dynamics models use a graph-theoretic abstraction to represent
the {\em attack-defense} structure, and use {\em parameters} to represent attack and defense ecapabilities.
This suggests us to classify MTD techniques based on what they will induce changes to the
attack-defense structure and/or parameters.

\subsubsection*{{\bf Networks-based MTD Techniques (Class I)}}

Example techniques that fall into this class are {\em IP address (and TCP port) randomization} and {\em dynamic access control}.
The basic idea underlying IP address and TCP port randomization is to frequently shuffle the IP addresses of computers by using various methods.
One method is to use virtual machine techniques, such as
migrating ensembles of virtual machines \cite{Keller:2012:LME:2390231.2390250}
and others \cite{IBMVirtualWire2013,DBLP:series/ais/YackoskiBYL13}.
Another method is to use networking techniques,
such as Network Address Space Randomization (NASR) whereby
IP addresses can be dynamically assigned (in lieu of DHCP) to render the attacker's hitlist useless \cite{Antonatos:2005:DAH:1103626.1103633},
IP hopping \cite{IPHopping} and others \cite{BBN-DYNAT}.
A recent variant also considers constraints and how to minimize the operation cost \cite{Jafarian:2012:ORH:2342441.2342467}.

The basic idea underlying dynamic access control is to dynamically regulate which computers or which network address space can
directly have access to the services in which other network address space.
For example, certain servers on a campus network only accommodate service requests from certain classrooms.
By dynamically randomizing IP addresses of authorized computers (e.g., using aforementioned techniques),
some compromised computers cannot launch direct attacks against some target computers.

\subsubsection*{{\bf Hosts-based MTD Techniques (Class II)}}

Four kinds of techniques fall into this class: instruction-level, code-level, memory-level, and application-level.
One {\em instruction-level} technique is called Instruction Set Randomization (ISR),
which aims to randomize the instructions of each process so that the attacker cannot inject executable malicious code
\cite{Kc:2003:CCA:948109.948146,Barrantes:2003:RIS:948109.948147}.
ISR uses a program-specific key to encrypt the instructions of a program
and the processor uses the same key to decrypt and run the instructions,
where encryption is often based on binary transformation tools, and decryption is often based on
dynamic binary instrumentation tools \cite{Barrantes:2003:RIS:948109.948147,Barrantes:2005:RIS:1053283.1053286},
emulators \cite{Kc:2003:CCA:948109.948146,Boyd:2010:GAI:1850488.1850601}, or
architectural support \cite{Papadogiannakis:2013:AAS:2508859.2516670,Sovarel:2005:WFE:1251398.1251408,Weiss:2006:KKA:1191820.1191890}.

One {\em code-level} technique is {\em code randomization} \cite{Cohen:1993:OSP:179007.179012,Forrest:1997:BDC:822075.822408}.
Code randomization offers fine-grained protection against code reuse attacks by
substituting/reordering instructions, inserting NOPs, and re-allocating registers.
Code randomization operations can be conducted at the compiler  \cite{Giuffrida:2012:EOS:2362793.2362833,series/ais/JacksonSHMWGBWF11} or
virtual machine level \cite{HiserOakland12}, or via static binary rewriting \cite{DBLP:conf/sp/PappasPK12,Wartell:2012:BSS:2382196.2382216}
and runtime binary rewriting \cite{Bruening:2003:IAD:776261.776290,Kiriansky:2002:SEV:647253.720293,Luk:2005:PBC:1065010.1065034,Nethercote:2007:VFH:1250734.1250746}.
Dynamically generated code can be randomized as well \cite{Homescu:2013:LTC:2541806.2516675}.

One {\em memory-level} technique is called Address Space Layout Randomization (ASLR), which
defeats code-injection attacks by randomizing the memory layout of
a program (e.g., stack) either at the compile time or at the runtime \cite{ASLR}.
ASLR can protect an executable (including the associated static data, stack, heap and dynamic libraries) \cite{Bhatkar:2003:AOE:1251353.1251361}
and the operating system kernel \cite{Giuffrida:2012:EOS:2362793.2362833},
but cannot defeat code reuse attacks.

One {\em application-level} technique is called $N$-version programming \cite{AvizienisNVersionProgramming1985},
by which the defender can dynamically use different implementations of the same program function.
Another technique is called {\em proactive cryptography}.
Cryptographic properties proven in abstract models are undermined by
attacks (e.g., malware) that can compromise cryptographic keys.
Threshold cryptography can avoid this single-point-of-failure
because it ``split" a key into $m$ pieces such that
compromising fewer than $t$ pieces will not cause exposure of the key, while the cryptographic function can
be computed when any $t$ of the $m$ pieces participate \cite{S79,DF89}.
Proactive cryptography can render the compromised pieces of a key useless once the pieces are re-shuffled \cite{HJJKY97}.

\subsubsection*{{\bf Instruments-based MTD Techniques (Class III)}}

The defender can utilize honeypot-like techniques to capture new attacks.
However, the attacker can ``tomograph" honeypots and bypass the IP addresses monitored by them.
Therefore, the defender can dynamically change the IP addresses monitored by honeypots \cite{CaiHoneypot2009}.

\subsection{Cyber Epidemic Dynamics Model Accommodating MTD}
\label{sec:model}

\subsubsection*{{\bf Cyber Epidemic Dynamics Models}}
The basic idea underlying cyber epidemic dynamics models (see, for example, \cite{WangSRDS03,WangTISSEC08,TowsleyInfocom05,VanMieghemIEEEACMTON09,XuTDSC12,XuTAAS12,XuTAAS2010})
can be explained as follows.
Cyber attacks are often launched from compromised computers against vulnerable computers.
This means that there exists an $(attacker,victim)$ relation, which captures that an attacker (e.g., compromised computer)
can {\em directly} attack a victim (e.g., vulnerable) computer.
In the extreme case where any computer can attack any other computer, this relation induces a complete graph structure.
In general, any graph structure can be relevant.
The resulting graph structures are called {\em attack-defense structures},
where compromised computers/nodes may be detected and cured, but may later get attacked again.
Such models can naturally abstract attack and defense capabilities as {\em parameters} that are associated to the nodes and edges of
attack-defense structures.
A core concept in cyber epidemic dynamics models is the so-called {\em epidemic threshold}, namely
a {\em sufficient condition} under which the epidemic dynamics converges to the clean state.

\subsubsection*{{\bf Accommodating MTD}}
We adapt the cyber epidemic dynamics model introduced in \cite{XuTAAS2010}, which
considers {\em reactive adaptive defense},
to accommodate MTD (i.e., {\em proactive defense}).
Specifically, the afore-discussed {\bf Class I} MTD techniques can be accommodated with {\em dynamic attack-defense structures},
because they can cause that an infected computer may be able to attack a vulnerable computer at time $t_1$ but not at time $t_2>t_1$
(e.g., the vulnerable computer's IP address has been changed).
{\bf Class II} MTD techniques can be accommodated with
dynamic parameters because they can affect capabilities of attacker and defender over time.
{\bf Class III} MTD techniques can be accommodated with dynamic attack-defense structures (because an IP address assigned to honeypot
at time $t_1$ may be re-assigned to a production computer at time $t_2>t_1$)
{\em and} dynamic parameters (because the defender could learn zero-day attacks
from honeypot-captured data, and identify and disseminate countermeasures to prevent/detect such attacks).
As such, our characterization study can accommodate the three classes of MTD techniques.

Specifically, we consider a cyber epidemic dynamics model with dynamic attack-defense structure $G(t)=(V,E(t))$,
where $V$ is the set of nodes (e.g., computers) and $E(t)$ is the set of edges at time $t$
such that $(w,v)\in E(t)$ means that node $w$ can attack node $v$ at time $t$.
Suppose $|V|=n$. We may think $V=\{1,\ldots,n\}$ as well.
Let $A(t)=[A_{vu}(t)]$ denote the  adjacency matrix of $G(t)$, where $A_{vu}(t)=1$ if
$(u,v)\in E(t)$ and $A_{vu}(t)=0$ otherwise.
Naturally, we have $A_{uu}(t)=0$ for all $u\in V$ (i.e., a computer does not attack itself).
Suppose any node $v\in V$ has two possible states: {\em secure} or {\em infected}.
A node $v\in V$ is {\em secure} if it is vulnerable but not successfully attacked yet, and \infected\ if it is successfully attacked.
Let $i_v(t)$ and $s_v(t)$ respectively be the probabilities that $v\in V$ is {\em infected} and {\em secure} at time $t$, where $i_v(t)+s_v(t)=1$.

Let $\gamma(t)$ be the probability that at time $t$, an \infected\ node $u\in V$ successfully attacks a \secure\ node $v\in V$ over $(u,v)\in E(t)$.
Let $\beta(t)$ be the probability that an \infected\ node $v\in V$ becomes \secure\ at time $t$.
Suppose the attacks are independently launched. The probability that a \secure\ node $v\in V$ becomes \infected\ at time $t$ is \cite{XuTAAS2010}:
\begin{align*}
\xi_v(t)=1-\prod_{G(t)=(V,E(t)):A_{vu}(t)=1}\left(1-i_u(t)\gamma(t) \right).
\end{align*}
The master dynamics equation is \cite{XuTAAS2010}:
\begin{eqnarray*}
\frac{d i_v(t)}{dt}=\xi_v(t)(1-i_v(t))-\beta(t) i_v(t) = \nonumber\\
\left(1-\prod\limits_{u\in V}(1- A_{vu}(t)i_u(t) \gamma(t))\right)(1-i_v(t))-i_v(t)\beta(t).
\end{eqnarray*}
This is the starting point of our study.
Table \ref{table:notations} lists the main notations used in the paper.

\begin{table}[h]
\centering
\caption{Main notations used throughout the paper\label{table:notations}}
\begin{tabular}{|r|p{.355\textwidth}|}
\hline
$X(t),X$ & $X(t)$ is a function of time $t$, while $X$ is not \\
$G(t)$ & $G(t)=(V,E(t))$ is attack-defense structure at time $t$: a graph of node set $V$ and edge set $E(t)$, where $|V|=n$ \\
$A(t)$ & adjacency matrix of $G(t)=(V,E(t))$ \\
$s_v(t),i_v(t)$ & the probability node $v$ is \secure\ or \infected\ at time $t$\\
$i_v^*$  & the probability node $v$ is \infected\ as $t\to \infty$ (if existing)\\
$I^*$    & $I^*\stackrel{def}{=}(i_1^*,\ldots,i_{n}^*)$ where $n=|V|$ \\
$\beta(t)$ & the {\em cure probability} that an \infected\ node becomes \secure\ at time $t$ (reflecting defense power)\\
$\gamma(t)$ & the {\em infection probability} that \infected\ node $u$ successfully attacks \secure\ node $v$ over edge $(u,v)\in E(t)$\\
$\C(t)$ & $\C(t)=(G(t),\beta(t),\gamma(t))$ is system configuration or posture at time $t$ \\
$\C_1$ & $\C_1=(G_1,\beta_1,\gamma_1)$ is the undesired/insecure configuration that violates the convergence condition\\
$\C_j$ & $\C_j =(G_j,\beta_j,\gamma_j)$ for $j \geq 2$ are MTD-induced desired configurations that satisfy the convergence condition \\
\hline
$\lambda_{1} (A)$ & the largest eigenvalue (in modulus) of matrix $A$\\
$I_n$ & the $n$-dimensional identity matrix \\
$\|\cdot\|$ & the 2-norm of vector or matrix\\
\hline
$s\leftarrow_R S$ & select $s$ as a random element of set $S$ \\
$T\leftarrow \EXP(a)$ & assign $T$ a value according to the exponential distribution with parameter $a$ \\
\hline
\end{tabular}
\end{table}

\subsection{Measuring the Power of MTD}
\label{explanation}

Let $I^*=(0,\ldots,0)$ denote the clean state or equilibrium $i_v^*\stackrel{def}{=}\lim_{t\to\infty}i_v(t)=0$ for all $v\in V$,
namely that there are no \infected\ computers in the equilibrium (i.e., the spreading dies out).
Cyber epidemic threshold is a sufficient condition under which the dynamics converges to $I^*=(0,\ldots,0)$.
In the special case that $G=(V,E)$ and $(\beta,\gamma)$ are independent of $t$,
it is known \cite{WangTISSEC08,TowsleyInfocom05,VanMieghemIEEEACMTON09,XuTAAS12} that the dynamics converges to $I^*=(0,\ldots,0)$ if
\begin{equation}
\label{eqn:static-epidemic-threshold}
\mu\stackrel{def}{=}\beta-\gamma\lambda_1(A)>0,
\end{equation}
where $\lambda_1(A)$ is the largest (in modulus) eigenvalue of the adjacent matrix $A$ of $G$.
If $\mu<0$, the dynamics does not converge to $I^*=(0,\ldots,0)$ at least for some initial values.

Suppose the defender is confronted with configuration or posture $\C_1=(G_1=(V,E_1), \beta_1,\gamma_1)$,
under which condition (\ref{eqn:static-epidemic-threshold}) does not hold.
Suppose the defender can launch combinations of MTD techniques to induce
configurations $\C_j=(G_j=(V,E_j), \beta_j,\gamma_j)$ for $j\geq 2$,
each of which satisfies condition (\ref{eqn:static-epidemic-threshold}).
If the defender can always assure $(G(t),\beta(t),\gamma(t))=\C_j=(G_j,\beta_j,\gamma_j)$ for {\bf any} $t>0$ and some $j\geq 2$,
the problem is solved because the defender can make the dynamics converge to $I^*=(0,\ldots,0)$ by launching MTD to induce $\C_j$.
However, it would be more realistic that the defender can maintain such configurations as $\C_j$ ($j\geq 2$) for a small period of time,
because the attacker can introduce (for example) zero-day attacks to force the system to depart from configuration $\C_j$ and enter configuration $\C_1$.
Moreover, the system may have to stay in configuration $\C_1$ at least for some period of time
because $G_1=(V,E_1)$ is necessary for facilitating some applications.
Figure \ref{fig:switching-illustration} illustrates the idea of using MTD to make the overall dynamics converge to $I^*=(0,\ldots,0)$,
while allowing the system to stay for some portion of time in the undersired configuration $\C_1$, which violates condition \eqref{eqn:static-epidemic-threshold}.

\begin{figure}[!htbp]
\centering
\includegraphics[width=0.5\textwidth]{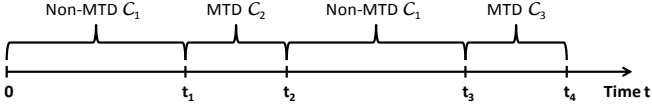}
\caption{Illustration of MTD-induced switching of configurations:
The system is in configuration $\C_1$ during time interval $[0,t_1)$,
in $\C_2$ during time interval $[t_1,t_2)$ because the defender launches MTD,
etc. Although $\C_1$ violates condition \eqref{eqn:static-epidemic-threshold},
the {\em overall dynamics} can converge to $I^*=(0,\ldots,0)$ because of MTD.
Note that $\C_2$ and $\C_3$ may reside in between two $\C_1$'s
(i.e., launching two combinations of MTD to induce $\C_2$ and $\C_3$ before returning to $\C_1$).}
\label{fig:switching-illustration}
\end{figure}

The preceding discussion leads us to define two measures of power of MTD.
The first definition captures the {\em maximum} time the system can afford to stay in the undersired configuration $\C_1$
while the overall dynamics converges to $I^*=(0,\ldots,0)$ because of MTD, without considering cost.

\begin{definition}
\label{definition-measure-without-cost}
\emph{{\bf (power of MTD without considering cost)}}
Consider undesired configuration $\C_1=(G_1,\beta_1,\gamma_1)$ that violates condition (\ref{eqn:static-epidemic-threshold}).
Suppose the defender can launch MTD to induce configurations $\C_j=(G_j,\beta_j,\gamma_j)$
for $j\in [2,\ldots,J]$, where each $\C_j$ satisfies condition (\ref{eqn:static-epidemic-threshold}).
Denote by $\mu_k=\beta_k-\gamma_k\lambda_1 (A_k)$ for $k=1,\ldots,J$, where $A_k$ is the adjacency matrix of $G_k$.
We say MTD is $(\mu_1,\mu_2,\ldots,\mu_J, \pi^*_1)$-powerful if it can make the overall dynamics
converge to $I^*=(0,\ldots,0)$, while allowing the system to stay in configuration $\C_1$ for the maximum $\pi^*_1$-portion of time in the equilibrium.
\end{definition}

The second definition captures the {\em minimum} cost with respect to a given portion of time, $\pi_1$, during which the system must stay in
configuration $\C_1$.

\begin{definition}
\label{definition-measure-with-cost}
\emph{{\bf (power of MTD while considering cost)}}
Consider undesired configuration $\C_1=(G_1,\beta_1,\gamma_1)$ that violates condition (\ref{eqn:static-epidemic-threshold}), and the potion of time $\pi_1$
that the system must stay in configuration $\C_1$.
Suppose the defender can launch MTD to induce configurations $\C_j=(G_j,\beta_j,\gamma_j)$
for $j=2,\ldots,J$, where each $\C_j$ satisfies condition (\ref{eqn:static-epidemic-threshold}).
Denote by $\mu_j=\beta_j-\gamma_j\lambda_1 (A_j)$ for $j=2,\ldots,J$, where $A_j$ is the adjacency matrix of $G_j$.
Consider cost function $h(\cdot):{\mathbb R}^+\to {\mathbb R}^+$
such that $h(\mu_j)$ is the cost of launching MTD to induce configuration $\C_j$ for $j=2,\ldots,J$, where $h'(\mu)\geq 0$ for $\mu>0$.
For give cost function $h(\cdot)$, we say MTD is $(\mu_1,\mu_2,\cdots,\mu_J,\pi_1,\Upsilon)$-powerful if the overall dynamics converges to $I^*=(0,\ldots,0)$
at the minimum cost $\Upsilon(\pi_2^*,\cdots,\pi_{J}^*)$, where
$\pi_j^*$ ($2\leq j \leq J$) is the portion of time the system stays in configuration $\C_j$ in the equilibrium.
\end{definition}

{\bf Remark}.
Definitions \ref{definition-measure-without-cost}-\ref{definition-measure-with-cost}
characterize the power of MTD from two complementary perspectives.
Definition \ref{definition-measure-without-cost} does not explicitly mention $\pi_2^*,\ldots,\pi_J^*$
because the problem of computing $\pi_2^*,\ldots,\pi_N^*$ is orthogonal to the existence of $\pi_1^*$.
Nevertheless, all of our results allow to explicitly compute $\pi_2^*,\cdots,\pi_N^*$.
Definition \ref{definition-measure-with-cost} explicitly mentions $\pi_2^*,\ldots,\pi_J^*$
because they are essential to the definition of minimum cost,
where $\pi_{1}$ is not a parameter of the cost $\Upsilon$ because the system must stay in $\C_1$ for a predetermined portion of time $\pi_1$.

\section{Power of MTD Inducing Dynamic Parameters}
\label{sdvsss}

In this section we characterize the power of MTD
that induces dynamic parameters but keeps the attack-defense structure intact (i.e., $G$ is independent of time $t$ throughout this section).
Let $A$ be the adjacency matrix of $G$.
We first recall the following theorem from \cite{XuTAAS2010} and present a corollary of it.

\begin{theorem}\label{sd_theorem}
\emph{(\cite{XuTAAS2010})}
Consider configurations $(G,\beta(t),\gamma(t))$, \\
where the dynamic parameters $\beta(t)$ and $\gamma(t)$ are driven by some ergodic stochastic process.
Let $\mathbb E(\beta(0))$ and $\mathbb E(\gamma(0))$ be the respective expectations of the stationary distributions of the process. Suppose convergences
$\lim_{t\to\infty}\int_{t_{0}}^{t_{0}+t}\beta(\tau)d\tau=\mathbb E(\beta(0))$ and \\
$\lim_{t\to\infty}\int_{t_{0}}^{t_{0}+t}\gamma(\tau)d\tau=\mathbb E(\gamma(0))$
are uniform with respect to $t_{0}$ almost surely.
If $\mathbb E(\beta(0))/\mathbb E(\gamma(0))>\lambda_1(A)$,
the dynamics converges to $I^*=(0,\ldots,0)$ almost surely; if
$\mathbb E(\beta(0))/\mathbb E(\gamma(0))<\lambda_1(A)$,
there might exist infected nodes in the equilibrium.
\end{theorem}

\begin{corollary}
\label{corollary:single-configuration}
Consider configurations $(G,\beta(t),\gamma(t))$, where $(\beta(t),\gamma(t))$ are driven by a homogeneous Markov process $\eta_t$
with steady-state distribution $[\pi_1,\cdots,\pi_{N}]$ and support \\
$\{(\beta_1,\gamma_1),\ldots,(\beta_N,\gamma_N)\}$,
meaning $\mathbb E(\beta_{\eta_t})=\pi_1\beta_1+\cdots+\pi_N\beta_N$ and $\mathbb E(\gamma_{\eta_t})=\pi_1\gamma_1+\cdots+\pi_N\gamma_N$.
If
\begin{align} \label{corollary:ave1}
\frac{\pi_1\beta_1+\cdots+\pi_N\beta_N}{\pi_1\gamma_1+\cdots+\pi_N\gamma_N}>\lambda_1(A),
\end{align}
the dynamics will converge to $I^*=(0,\ldots,0)$; if
\begin{align*}
\frac{\pi_1\beta_1+\cdots+\pi_N\beta_N}{\pi_1\gamma_1+\cdots+\pi_N\gamma_N}<\lambda_1(A),
\end{align*}
the dynamics will not converge to $I^*=(0,\ldots,0)$ at least for some initial value scenarios.
\end{corollary}

\subsection{Characterizing Power of MTD without Considering Cost}

In this case, despite that $\C_1=(G,\beta_1,\gamma_1)$ violates condition (\ref{eqn:static-epidemic-threshold}),
the system needs to stay as much as possible in configuration $\C_1$.
Fortunately, the defender can exploit MTD to make the overall dynamics converge to $I^*=(0,\ldots,0)$.
This is possible because MTD can induce configurations $\C_j=(G,\beta_j,\gamma_j)$ for $j=2,\ldots,N$,
where each $\C_j$ satisfies condition (\ref{eqn:static-epidemic-threshold}).
Denote by $\mu_j=\beta_j-\gamma_j\lambda_1(A)$ for $j=1,\cdots,N$.
Without loss of generality, suppose $\mu_1<0<\mu_2<\cdots<\mu_N$.
According to Corollary \ref{corollary:single-configuration}, if
\begin{align}\label{without-cost1}
\pi_1\mu_1+\cdots+\pi_N\mu_N>0,
\end{align}
then the dynamics will converge to $I^*=(0,\ldots,0)$.
Since inequality (\ref{without-cost1}) is strict and is a linear function of $\pi_1$,
in order to reach the maximum $\pi_1^*$
we need to introduce a sufficiently small constant $0<\delta\ll1$ and replace condition (\ref{without-cost1}) with
\begin{align}\label{without-cost2}
\pi_1\mu_1+\cdots+\pi_N\mu_N\geq\delta.
\end{align}
Theorem \ref{theorem:sd-ss} constructively identifies the maximum $\pi_1^*$, the maximal portion of time the system can afford to stay in $\C_1$.

\begin{theorem}
\label{theorem:sd-ss}
Suppose configuration $\C_1=(G,\beta_1,\gamma_1)$  violates condition (\ref{eqn:static-epidemic-threshold}).
Suppose MTD-induced configurations $\C_j=(G,\beta_j,\gamma_j)$ for $j=2,\cdots,N$ satisfy condition (\ref{eqn:static-epidemic-threshold})
as $0<\mu_2<\cdots<\mu_N$.
The maximal potion of time the system can afford to stay in configuration $\C_1$ is
\begin{equation}
\label{eq:pi-1-*}
\pi_1^*=\frac{\mu_N-\delta}{\mu_N-\mu_1},
\end{equation}
while the system will stay in configurations $\C_2,\ldots,\C_N$ respectively with portions of time given by
\begin{eqnarray*}
\pi^{*}_{2}=\cdots=\pi^{*}_{N-1}=0,\quad\pi^{*}_{N}=\frac{\delta-\mu_{1}}{\mu_{N}-\mu_{1}}.
\end{eqnarray*}
In other words, MTD is $(\mu_1,\cdots,\mu_{N},\pi_1^*)$-powerful.
\end{theorem}

\begin{proof}
Eq. (\ref{without-cost2}) implies
\begin{align}\label{without-cost-proof1}
\pi_1&\leq \frac{\pi_2\mu_2+\cdots+\pi_N\mu_N-\delta}{-\mu_1}\leq\frac{(\pi_2+\cdots+\pi_N)\mu_N-\delta}{-\mu_1}\\
\nonumber&=\frac{(1-\pi_1)\mu_N-\delta}{-\mu_1},
\end{align}
which means
\begin{align}\label{without-cost-proof2}
\pi_1\leq\frac{\mu_N-\delta}{\mu_N-\mu_1}.
\end{align}
Moreover, this maximum $\pi_1^*$ can be reached
if all the equalities in Eqs. (\ref{without-cost-proof1}) and (\ref{without-cost-proof2}) hold, namely
\begin{eqnarray*}
\pi_{2}=\cdots=\pi_{N-1}=0,\quad\pi_{N}=\frac{\delta-\mu_{1}}{\mu_{N}-\mu_{1}}.
\end{eqnarray*}
This means that the defender only needs to launch the MTD that induces configuration $\C_N=(G,\beta_{N},\gamma_{N})$.
\end{proof}

Theorem \ref{theorem:sd-ss} says that although the defender can launch MTD to induce a set of $(N-1)$ configurations,
$\C_2,\ldots,\C_N$ with $0<\mu_2<\cdots<\mu_N$,
only $\C_N$ matters.
This means that $\mu_k$ is indicative of the capability of a configuration.
Figure \ref{fig:pi-1-dependence} plots the dependence of $\pi_1^*$ on $-\mu_1$ and $\mu_N$ with $\delta=10^{-5}$.
We observe that for fixed $\mu_1$, the maximum portion of time $\pi_1^*$ monotonically non-decreases in $\mu_N$.
For example, by fixing $\mu_1=-0.4$, $\pi_1^*$ is a non-decreasing curve in $\mu_N$.

\begin{figure}[H]
\centering
\includegraphics[width=0.46\textwidth]{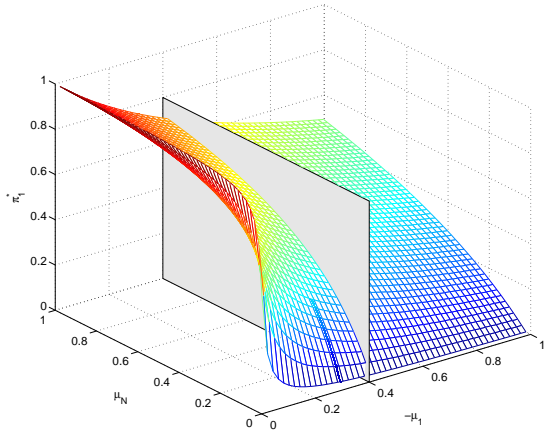}
\caption{Dependence of $\pi_1^*$ on $-\mu_1$ and $\mu_N$.
}
\label{fig:pi-1-dependence}
\end{figure}

\subsection{Characterizing Power of MTD while Considering Cost}
\label{sec_cost}

In this case, configuration $\C_1=(G,\beta_1,\gamma_1)$ is given and the time the system must stay in $\C_1$ is predetermined as $\pi_1$.
The defender wants to deploy MTD to make the overall dynamics converge to $I^*=(0,\ldots,0)$,
while minimizing the cost of launching MTD.
Denote the MTD-induced configurations by $\C_j=(G,\beta_j,\gamma_j)$ for $2\leq j\leq N$.
Note that cost may only be considered for $N\geq 3$
because when $N=2$, it is more meaningful to maximize $\pi_1^*$ (i.e., the preceding case).

Since we have proved that $\pi_1\leq \pi_1^* = \frac{\mu_N-\delta}{\mu_N-\mu_1}$ is necessary to make the dynamics converge to $I^*=(0,\ldots,0)$,
$\frac{\mu_N-\delta}{\mu_N-\mu_1}$ is the natural upper bound on $\pi_1$ (i.e., if $\pi_1$ is above the upper bound, we cannot assure
the dynamics will converge to $I^*=(0,\ldots,0)$, regardless of the cost).
Consider cost function $f(\cdot):{\mathbb R}^+\to {\mathbb R}^+$ as in Definition \ref{definition-measure-with-cost}.
such that $f(\mu_j)$ is the cost of launching MTD to induce configuration $\C_j$ for $2\leq j\leq N$, where $f'(\mu)\geq 0$ for $\mu>0$.
The objective is to minimize, for given cost function $f(\cdot)$ and any constant $f(\mu_1)$, the following cost:

\begin{align}\label{objective-function}
\Phi(\pi_{2},\cdots,\pi_{N})=\pi_1 f(\mu_1)+\sum_{j=2}^N\pi_j f(\mu_j)
\end{align}
subject to
\begin{align}\label{restrict_con1}
\frac{\pi_1 \beta_1+\sum_{j=2}^N\pi_j\beta_j}{\pi_1 \gamma_1+\sum_{j=2}^N\pi_j\gamma_j}>\lambda_1(A),~\pi_1 +\sum_{j=2}^N\pi_j=1,
~\pi_j\geq 0
\end{align}
Since $\frac{\pi_1 \beta_1+\sum_{j=2}^N\pi_j\beta_j}{\pi_1 \gamma_1+ \sum_{j=2}^N\pi_j\gamma_j}>\lambda_1(A)$ is equivalent to
$\pi_1 \mu_1+\sum_{j=2}^N\mu_j>0$,
Eq. (\ref{restrict_con1}) is equivalent to:
\begin{align}\label{restrict_con2}
\sum_{j=2}^N\pi_j{\mu_j}>-\pi_1 \mu_1,~\sum_{j=2}^N\pi_j=1-\pi_1, ~\pi_j\geq 0.
\end{align}
Since the objective is linear and the optimal solution would get on bound of the non-closed constraint (\ref{restrict_con2}),
we need to introduce a sufficiently small constant $0<\delta\ll 1$ and replace constraint (\ref{restrict_con2}) with
\begin{align}\label{restrict_con3}
\sum_{j=2}^N\pi_j{\mu_j}\geq-\pi_1 \mu_1+\delta,
~\sum_{j=2}^N\pi_j=1-\pi_1 ,
~\pi_j\geq 0.
\end{align}
Theorem \ref{theorem:sd-ss-minimum-cost}
shows how to find the minimum cost $\Phi(\pi_{2}^*,\cdots,\pi_{N}^*)$ according to constraints (\ref{objective-function}) and (\ref{restrict_con3}),
and therefore gives an algorithm for the optimization problem.
Proof of Theorem \ref{theorem:sd-ss-minimum-cost} is deferred to the Appendix.

\begin{theorem}
\label{theorem:sd-ss-minimum-cost}
Suppose configuration $\C_1=(G,\beta_1,\gamma_1)$ violates condition (\ref{eqn:static-epidemic-threshold}).
Suppose MTD-induced configurations $\C_j=(G,\beta_j,\gamma_j)$ for $j=2,\cdots,N$ satisfy condition (\ref{eqn:static-epidemic-threshold}).
Suppose $\pi_1 $, where $0<\pi_1 \leq \frac{\mu_N-\delta}{\mu_N-\mu_1}$, is the potion of time the system must stay in $\C_1$.
Suppose $f(\cdot)$ is the cost function as discussed above.
Define
\begin{align}\label{sdvsss-k*}
\mu_{k^*}=\min\left\{\mu_k|\mu_k>\frac{-\pi_1 \mu_1}{(1-\pi_1 )},~2\leq k \leq N\right\}
\end{align}
 and for $2\leq l <m\leq N$,
\begin{align}\label{sdss-cost}
\nonumber F(\mu_l,\mu_m)=&\pi_1 f(\mu_1)+\frac{f(\mu_m)-f(\mu_l)}{\mu_m-\mu_l}(\delta-\pi_1 \mu_1)\\
&+\frac{\mu_mf(\mu_l)-\mu_lf(\mu_m)}{\mu_m-\mu_l}(1-\pi_1 ).
\end{align}
If $k^*=2$, the minimal cost under constraint (\ref{restrict_con3}) is
\begin{align*}
\min\limits_{\pi_{2},\cdots,\pi_{N}}\Phi(\pi_{2},\cdots,\pi_{N})=\pi_1 f(\mu_1)+(1-\pi_1 )f(\mu_2),
\end{align*}
which is reached by launching MTD to induce configuration $\C_{2}$ only.
If $k^*>2$, the minimal cost under constraint (\ref{restrict_con3}) is
\begin{align}\label{sdss-minimum-cost}
\min\limits_{\pi_{2},\cdots,\pi_{N}}\Phi(\pi_{2},\cdots,\pi_{N})=\min_{l< k^*\leq m}F(\mu_l,\mu_m).
\end{align}
Denote by $\{\mu_{l^*},\mu_{m^*}\}=\arg\min\limits_{l< k^*\leq m}F(\mu_l,\mu_m)$. The minimal cost
is reached by launching MTD to induce configurations $\C_{l^*},\C_{m^*}$ respectively with portions of time $[\pi_{l^*},\pi_{m^*}]$:
\begin{align}\label{eqn:sdvsss-deploy}
\left[\begin{array}{cc}
\pi_{l^*}\\
\pi_{m^*}
\end{array}\right]=\frac{1}{\mu_{m^*}-\mu_{l^*}}\left[\begin{array}{cc}
(\mu_{m^*}-\delta)+\pi_1 (\mu_1-\mu_{m^*})\\
-(\mu_{l^*}-\delta)+\pi_1 (\mu_{l^*}-\mu_1)
\end{array}\right],
\end{align}
where $0<\delta\ll 1$ is some constant.
That is, MTD is \\
 $(\mu_1,\mu_2,\cdots,\mu_N,\pi_1 ,\Phi)$-powerful.
\end{theorem}

\subsection{Algorithm for Orchestrating Optimal MTD}

When not considering cost, Theorem \ref{theorem:sd-ss} constructively gives a method for optimally launching MTD.
When considering {\em arbitrary} cost function $f(\cdot)$, Theorem \ref{theorem:sd-ss-minimum-cost} constructively
shows how to find the minimum cost $\Phi(\pi_{2}^*,\cdots,\pi_{N}^*)$ according to constraints (\ref{objective-function}) and (\ref{restrict_con3}),
and therefore gives a method for computing the minimum cost and the corresponding strategy for optimally launching MTD.
Theorems \ref{theorem:sd-ss}-\ref{theorem:sd-ss-minimum-cost} suggest many possible ways/algorithms to achieve the goal,
with Algorithm \ref{algorithm:sd-ss-strategy} being a concrete example.

\begin{algorithm}[hbt]
\caption{Launching optimal MTD (dynamic parameters)\label{algorithm:sd-ss-strategy}}
INPUT: initial configuration $\C_1=(G,\beta_1,\gamma_1)$,
MTD-induced configurations $\C_j=(G, \beta_j,\gamma_j)$ for $j=2,\cdots,N$ and $N\geq 2$, constant $a>0$ determining time resolution,
optional cost function $f(\cdot)$, $\delta$ ($0<\delta<<1$), optional $\pi_1$ \\
OUTPUT: Optimal MTD strategy

\begin{algorithmic}[1]
\IF{cost function is not given (i.e., no need to consider cost)}
\STATE Compute $\pi_1^*$ according to Eq. (\ref{eq:pi-1-*})
\WHILE{TRUE}
\STATE Wait for time $T_1\leftarrow \EXP(a/\pi_1^*)$ \COMMENT{system in $\C_1$}
\STATE Launch MTD to make system stay in $\C_N$ for time $T_{N}\leftarrow \EXP(a/(1-\pi_1^*))$
\STATE Stop launching MTD \COMMENT{system returns to $\C_1$}
\ENDWHILE
\ELSE
\STATE Compute $k^*$ according to Eq. (\ref{sdvsss-k*})
\IF{$k^*>2$}
\STATE Compute $\mu_{l^*},\mu_{m^*}$ according to Eq. (\ref{sdss-cost})
\STATE Compute $\pi_{l^*},\pi_{m^*}$ according to Eq. (\ref{eqn:sdvsss-deploy})
\ELSE
\STATE Set $l^*=m^*=2$ and $\pi_{l^*}=1-\pi_1$
\ENDIF
\STATE Wait for time $T_1\leftarrow \EXP(a/\pi_1)$  \COMMENT{system in $\C_1$} 
\STATE $j \leftarrow_R \{l^*,m^*\}$ ~~~~\COMMENT{$j=2$ when $l^*=m^*=2$}
\WHILE{TRUE}
\STATE Launch MTD to make system stay in $\C_{j }$ for time $T\leftarrow \EXP(a/\pi_{j })$ \COMMENT{system in $\C_j$}
\STATE Set $\Delta= \{1,l^*,m^*\} - \{j \}$
\STATE $j \leftarrow_R \Delta$  ~~~~\COMMENT{$j=1$ when $l^*=m^*=2$}
\IF{$j=1$}
\STATE Stop launching MTD and wait for time $T_1\leftarrow \EXP(a/\pi_1)$  \COMMENT{system in $\C_1$} 
\STATE $j \leftarrow_R \{l^*,m^*\}$
\ENDIF
\ENDWHILE
\ENDIF
\end{algorithmic}
\end{algorithm}

Specifically, lines 2-7 describe the algorithm corresponding to Theorem \ref{theorem:sd-ss} (i.e., not considering cost), where
line 4 instructs the defender not to launch MTD so that the system stays in configuration $\C_1$ for a period of time $T_1$, and
line 5 instructs the defender to launch MTD to make the system stay in configuration $\C_N$ for the period of time $T_N$.
On the other hand, lines 9-26 describe the algorithm corresponding to Theorem \ref{theorem:sd-ss-minimum-cost} (i.e., considering cost).
If $k^*=2$, the defender needs to make the cyber system stay alternatively in configurations $\C_1$ and $\C_2$.
If $k^*>2$, the defender needs to make the cyber system stay alternatively in configurations $\C_1$ for a period of time $T_1$,
 in configuration $\C_{l^*}$ for a period of time $T_{l^*}$ and/or in configuration $\C_{m^*}$ for a period of time $T_{m^*}$.
Depending on the random coins flipped on lines 21 and 24, possible configuration sequences include:
$\C_1,\C_{l^*},\C_{m^*},\C_1,\ldots$ and $\C_1,\C_{l^*},\C_{m^*},\C_{l^*},\C_1,\ldots$.
Another algorithm for achieving the same goal it to make the system in $\C_1,\C_{l^*},\C_{m^*},\C_1,\C_{l^*},\C_{m^*},\ldots$
periodically for periods of time $T_1,T_{l^*},T_{m^*}$, respectively.

The computational complexity of Algorithm \ref{algorithm:sd-ss-strategy} is straightforward.
When not considering cost, the algorithm incurs $O(1)$ computational complexity.
When considering cost, line 9 incurs $O(N)$ complexity for searching $k^*$ according to \eqref{sdvsss-k*}, line 11 incurs
$O(N^2)$ complexity for searching the optimal $l^*$ and $m^*$ according to \eqref{sdss-cost},
and all other steps incur $O(1)$ complexity.

\subsection{Simpler Algorithm for Convex and Concave Cost Functions}

Algorithm \ref{algorithm:sd-ss-strategy} applies to {\em arbitrary} cost function $f(\cdot)$.
We make a further observation on Theorem \ref{theorem:sd-ss-minimum-cost}, which says that
for any given $0<\pi_1^*\leq \frac{\mu_N-\delta}{\mu_N-\mu_1}$ and cost function $f(\cdot)$,
if $k^*=2$, the minimum cost is reached by inducing configuration $\C_2$;
if $k^*>2$, Eqs. (\ref{sdss-cost}) and (\ref{sdss-minimum-cost}) indicate that the minimum cost is dependent upon the property of $f(\cdot)$.
Now we show that when $f(\cdot)$ is convex or concave, which may hold for most scenarios, we can obtain closed-form results on $l^{*},m^*$,
and therefore Algorithm \ref{algorithm:sd-ss-strategy} is naturally simplified.
Recall that $\mu_2<\cdots<\mu_N$.
Define $R(\mu_l,\mu_m)=\frac{f(\mu_m)-f(\mu_l)}{\mu_m-\mu_l}$. It can be verified that
\begin{align*}
F(\mu_l,\mu_m)=&\pi_1^*f(\mu_1)+(1-\pi_1^*)f(\mu_l)\\
&+R(\mu_l,\mu_m)[(\delta-\pi_1^*\mu_1)-\mu_l(1-\pi_1^*)]\\
=&\pi_1^*f(\mu_1)+(1-\pi_1^*)f(\mu_m)\\
&+R(\mu_l,\mu_m)[(\delta-\pi_1^*\mu_1)-\mu_m(1-\pi_1^*)].
\end{align*}

\begin{itemize}
\item If $f(\cdot)$ is convex, namely $f''(\cdot)\geq0$, then for fixed $\mu_l$ (or $\mu_m$), $R(\mu_l,\mu_m)$ is monotonically non-decreasing in $\mu_m$ (or $\mu_l$).
Note that $\mu_l< \frac{\delta-\pi_1^*\mu_1}{1-\pi_1^*}\leq \mu_m$, where $\delta\ll 1$.
For fixed $\mu_l$ (or $\mu_m$), $F(\mu_l,\mu_m)$ is monotonically non-decreasing (or non-increasing) in $\mu_m$ (or $\mu_l$). The minimum cost is
    \begin{align}\label{sdvsss-optimal-convex}
    \min\limits_{l< k^*\leq m}F(\mu_l,\mu_m)=F(\mu_{k^*-1},\mu_{k^*}).
    \end{align}
Having identified $\mu_{k^*-1},\mu_{k^*}$, one can compute $\pi_{k^*-1},\pi_{k^*}$ according to \eqref{eqn:sdvsss-deploy}.
Thus, lines 11 and 12 are simplified by this analytical result, with the complexity of searching for the optimal solution (i.e., $k^*$ in this case) reduced to $O(N)$.
\item If $f(\cdot)$ is concave, namely $f''(\cdot)\leq 0$, then
for fixed $\mu_l$ (or $\mu_m$), $R(\mu_l,\mu_m)$ is monotonically non-increasing (or non-decreasing) in $\mu_m$ (or $\mu_m$). The minimum cost is
\begin{align}\label{sdvsss-optimal-concave}
\min\limits_{l< k^*\leq m}F(\mu_l,\mu_m)=F(\mu_{2},\mu_{N}).
\end{align}
Similarly, having identified $\mu_{2},\mu_{N}$, one can compute $\pi_{2},\pi_{N}$ according to \eqref{eqn:sdvsss-deploy}.
Thus, lines 11 and 12 are simplified by this analytical result, with the complexity of searching for the optimal solution reduced to $O(1)$.
\end{itemize}
The above discussion suggests the following:
If $f(\cdot)$ is convex,
the defender only needs to launch MTD to induce configurations $\C_{k^*-1},\C_{k^*}$;
if $f(\cdot)$ is concave,
the defender only needs to launch MTD to induce configurations $\C_2,\C_N$.

To illustrate the influence of $f(\cdot)$ on the power of MTD, we set $N=4$,
$(\beta_1,\gamma_1)=(0.2,0.00422)$, $(\beta_2,\gamma_2)=(0.4,0.000845)$, $(\beta_3,\gamma_3)=(0.6,0.00169)$, $(\beta_4,\gamma_4)=(0.8,0.00169)$,
$\delta=10^{-5}$, $\lambda_1(A)=118.4$, $\mu_1-\delta\approx-0.3$, $\mu_2-\delta \approx0.3$, $\mu_3-\delta \approx0.4$, and $\mu_4-\delta \approx0.6$.
From Eq. (\ref{eq:pi-1-*}), we get $\pi_1^*\leq\frac{\mu_4-\delta}{\mu_4-\mu_1}=\frac{2}{3}$. We set $\pi_1^*=\frac{3}{5}$, which means $k^*=4$.
Figure \ref{fig:optimal-f} plots the total cost $\Phi$ with different cost functions $f(\cdot)$, where
$\pi_4=1-\pi_1-\pi_2-\pi_3$.
The shadow area in the $\pi_2\pi_3$-plane is the constrain slope of $\pi_2,\pi_3$ with respect to condition (\ref{restrict_con3}).
Note that $\Phi =\pi_1^*f(\mu_1)+\pi_2f(\mu_2)+\pi_3f(\mu_3)+(1-\pi_1^*-\pi_2-\pi_3)f(\mu_4)$,
which is linear non-increasing in $\pi_2$ for fixed $\pi_3$ (also in $\pi_3$ for fixed $\pi_2$).
For convex function $f(x)=100(x+0.1)^2$, the above analysis revealed that the minimum cost is reached at $[\pi_3,\pi_4]=[\pi_{l^*},\pi_{m^*}]$ as given by Eq.
(\ref{eqn:sdvsss-deploy}), namely by launching MTD to induce configurations $\C_3,\C_4$.
Figure \ref{fig:optimal-f1} shows that the minimum cost is reached at $[\pi_2,\pi_3,\pi_4]=[0,0.3,0.1]$ and the minimum cost is $14.6$,
which matches the analytic result given by Eq. \eqref{sdvsss-optimal-convex}.
For concave function $f(x)=10\sqrt{x+0.5}$,
Figure \ref{fig:optimal-f2} shows that
the minimum cost is reached at $[\pi_1,\pi_2,\pi_3]=[0.2,0,0.2]$ and the minimum cost is $6.5696$, which matches the analytic result given by Eq. \eqref{sdvsss-optimal-concave}.

\begin{figure}[!htbp]
\centering
\subfigure[Convex cost function $f(x)=100(x+0.1)^2$]{
\includegraphics[width=0.46\textwidth]{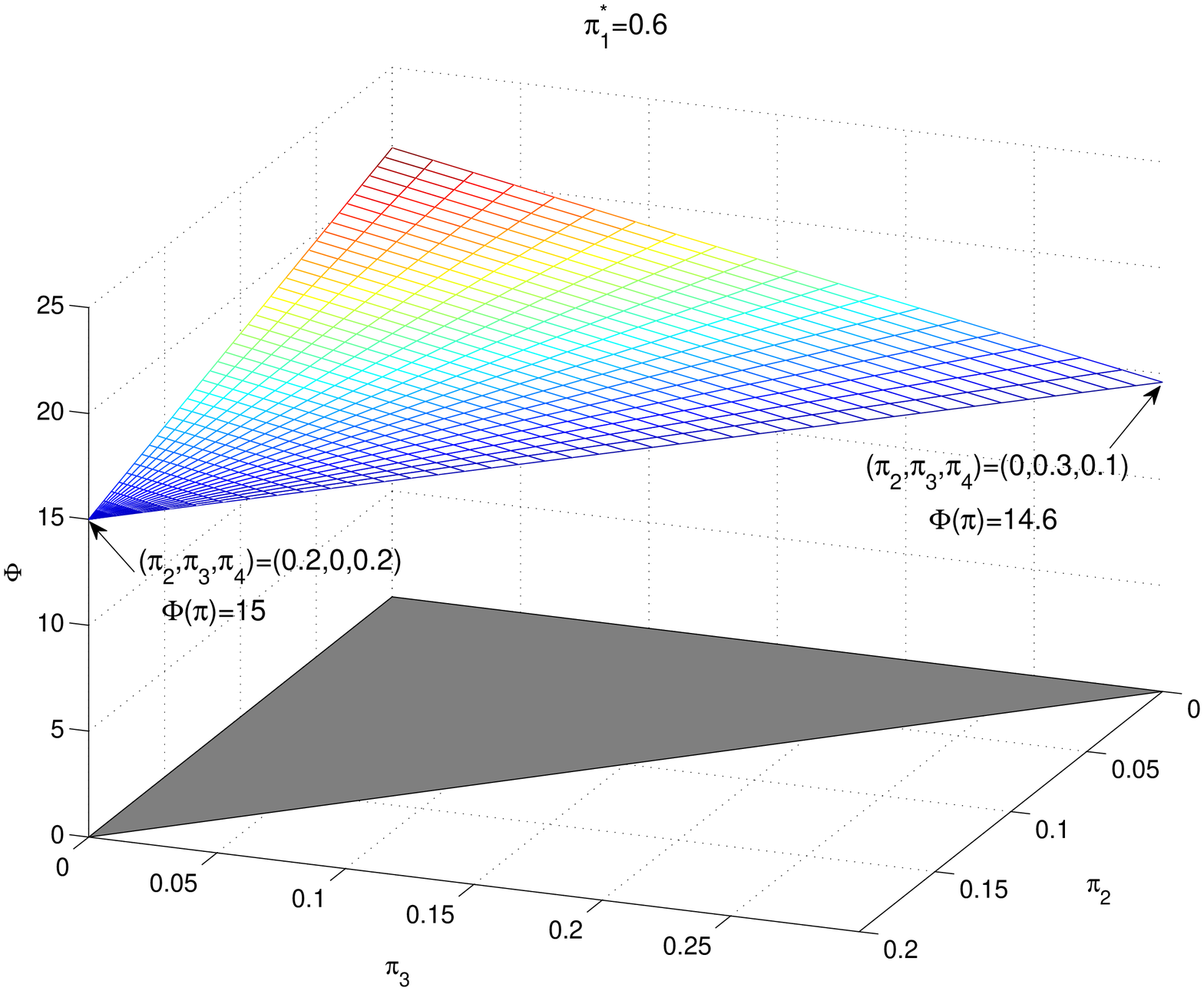}
\label{fig:optimal-f1}
}
\subfigure[Concave cost function $f(x)=10\sqrt{x+0.5}$]{
\includegraphics[width=0.46\textwidth]{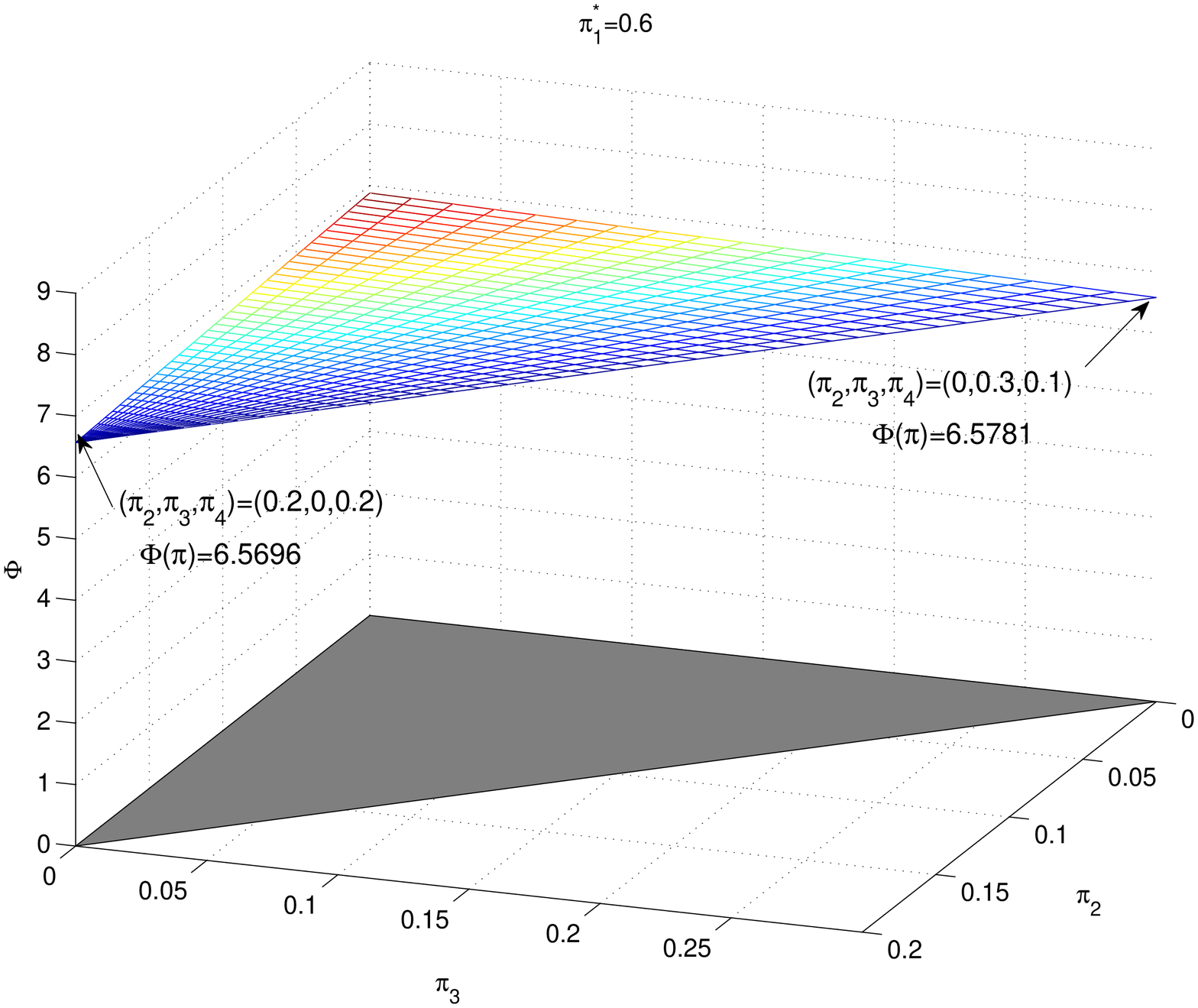}
\label{fig:optimal-f2}
}
\caption{Dependence of $\Phi$ on $\pi_1,\pi_2$ under different cost functions $f(\cdot)$}
\label{fig:optimal-f}
\end{figure}

\section{Power of MTD Inducing Dynamic Structures}
\label{dsvsss}

In this section, we characterize the power of MTD that induces
dynamic attack-defense structures $G(t)$, while the parameters $(\beta,\gamma)$ are kept intact.
More specifically, suppose configuration $\C_1=(G_1,\beta,\gamma)$ violates condition  \eqref{eqn:static-epidemic-threshold}.
Suppose MTD-induced configurations $\C_l=(G_l,\beta,\gamma)$ for $l=2,\cdots,N'$ and $N'\geq 2$ satisfy condition  \eqref{eqn:static-epidemic-threshold}.
We want to identify a Markov process strategy $\sigma_t$, defined over $\C_1,\C_2,\ldots,\C_{N'}$,
to make the dynamics converge to equilibrium $I^*=(0,\ldots,0)$,
while staying in configuration $\C_1$ as much as possible or minimizing the cost of launching MTD.
Throughout this section, let $A_l$ be the adjacency matrix of $G_l$ for $l=2,\ldots,N'$.

We start with a general result where one or more configurations violating condition \eqref{eqn:static-epidemic-threshold}.

\begin{theorem}\label{dsvsss-theorem}
Consider $\C_l=(G_l,\beta,\gamma)$ for $l=1,\cdots,N'$,
where $\C_\ell=(G_\ell,\beta,\gamma)$ for some $1\leq \ell \leq j$ violate condition  \eqref{eqn:static-epidemic-threshold}
but $\C_k=(G_k,\beta,\gamma)$ for some $j<k\leq N'$ satisfy condition  \eqref{eqn:static-epidemic-threshold}. Then, the overall dynamics converges to $I^*=(0,\ldots,0)$ almost surely under
Markov process strategy $\sigma_t$ with infinitesimal generator $Q=(q_{uv})_{N'\times N'}$ defined as:
\begin{enumerate}
\item[(i)]
for $k> j$,
$-q_{kk}\leq \frac{2a[\beta-\gamma \lambda_1(A_k)-\delta]}{\frac{jc+N'-1-j}{N'-1}-a}$;

\item[(ii)]
for $\ell\leq j$,
$-q_{\ell\ell}\geq \frac{2b[\gamma\lambda_{1}(A_\ell)-\beta+\delta]}{b-\frac{c(j-1)}{N'-1}-\frac{N'-j}{N'-1}}$;
\item[(iii)] $q_{rp}=\frac{-q_{rr}}{N'-1}$ for all $p\ne r$ and $p,r\in\{1,\ldots,N'\}$.
\end{enumerate}
\end{theorem}

\begin{proof}
Suppose (if needed, with reordering)
\begin{align}\label{dsvsss-theorem-order}
\lambda_1(A_1)\geq \cdots\lambda_1(A_j)>\frac{\beta}{\gamma}>\lambda_1(A_{j+1})\cdots\geq \lambda_{1}(A_{N'}).
\end{align}
For any $k> j$, $[\gamma A_k-\beta I_n]$ is a Hurwitz matrix \cite{Teschl} (i.e., real parts of all eigenvalues are negative), meaning that
there exist positive definite matrices $P_k<I_n$ and a constant $0<\delta\ll 1$ such 
that $(P_k[\gamma A_k-\beta I_n])^s=[\gamma \lambda_1(A_k)-\beta +\frac{\delta}{2}] P_k<0$. 
We can find positive definite matrices $P_\ell$ with $\ell\le j$ and positive constants $a<1<b<c$ such that
\begin{align*}
aI_n<P_k<I_n<b I_n<P_\ell<cI_n, ~\forall k> j,~\ell\leq j,
\end{align*}
and
\begin{align*}
\{P_\ell[\gamma A_\ell-\beta I_n]\}^s\leq [\gamma \lambda_1(A_\ell)-\beta+\frac{\delta}{2}]P_\ell,~\forall~\ell\leq j.
\end{align*}
By combining
(i)-(iii) in the condition of the theorem, we obtain
\ignore{
\begin{align}
\nonumber q_{mm}\bar{P}
+2\lambda_{\max}(\{P_m[\gamma A_m-\beta I_n]\}^s)
+2\delta<0, \forall m,
\end{align}
hence, for any $x\in\mathbb R^n$,
\begin{align*}
&x^{\top}\bigg\{\{P_m[\gamma A_m-\beta I_n]\}^s+\frac{1}{2}\sum_{r=1}^{N'} q_{rm} P_r+\delta I_n\bigg\}x\\
\leq &\frac{1}{2}\bigg\{ q_{mm}\bar{P}
+2\lambda_{\max}(\{P_m[\gamma A_m-\beta I_n]\}^s)
+2\delta\bigg\}x^{\top}x\\
\leq& 0,
\end{align*}
which implies }
\begin{align}\label{dsvsss-slow-con}
\{P_m[\gamma A_m-\beta I_n]\}^s+\frac{1}{2}\sum_{r=1}^{N'} q_{rm} P_r\leq-\frac{\delta}{2} I_n.
\end{align}
Since the parameters are static and the structures are driven by Markov process $\sigma_t$, the dynamics of $i_v(t)$ for $v\in V$ is:
\begin{align}\label{sys1}
\frac{d i_v(t)}{dt}=\bigg[1-\prod_{u\in V}[1-\gamma (A_{\sigma_t})_{vu} i_u(t)]\bigg](1-i_v(t))-\beta i_v(t).
\end{align}
Since
$$\bigg[1-\prod\limits_{u\in V}[1-\gamma (A_{\sigma_t})_{vu}i_u(t)]\bigg](1-i_v(t))\leq\sum\limits_{u\in V}\gamma (A_{\sigma_t})_{vu}i_u(t)$$
always holds, we define a new variable $y_v(t)$ with dynamics:
\begin{align}\label{compare}
\frac{d y_v(t)}{dt}=\sum_{u\in V}\gamma (A_{\sigma_t})_{vu}y_u(t)-\beta y_v(t).
\end{align}
Note that any sample point $w\in\Omega$ corresponds to a deterministic $\sigma_t(w)$.
Let
\begin{eqnarray*}
i(t)=[i_1(t),\cdots,i_n(t)]^{\top}, y(t)=[y_1(t),\cdots,y_n(t)]^{\top}
\end{eqnarray*}
be the solutions of systems (\ref{sys1}) and (\ref{compare}) under the Markov switching process $\sigma_t(w)$ respectively.
From the comparison theory of differential equations, we know $i(t)\leq y(t)$ holds if $i(0)=y(0)$, which implies that $\mathbb E [\|i(t)\|^2]\leq \mathbb E[\|y(t)\|^2]$.

Let $V(y(t),t,\sigma_t)=\frac{1}{2} y(t)^{\top}P(\sigma_t)y(t)$ and $\zeta=\frac{\delta}{2c}$.
The joint process $\{(y(t),\sigma_{t}):t>0\}$ is a strong Markov process
and the infinitesimal generator of the process is:
\begin{eqnarray*}
\mathcal{L}=Q+diag\{y^{\top}P^{\top}(1)\frac{\partial}{\partial y},
\cdots,y^{\top}P^{\top}(N')\frac{\partial}{\partial y}\}
\end{eqnarray*}
Then, we have
\begin{align*}
\mathcal{L}V(y,t,j)&=&\sum_{k=1}^{N}q_{kj}V(y,t,k)+(\frac{\partial
V(y,t,j)}{\partial y})^{\top}\dot{y},
\end{align*}
From the Dynkin Formula \cite{Dynkin} and Eq. (\ref{dsvsss-slow-con}), we have
\begin{eqnarray*}
&&\mathbb{E}e^{\zeta t}V(y(t),t,\sigma_t) \\
&=& V(y_0,0,\sigma_0)+\mathbb{E} \int_{0}^t \zeta e^{\zeta\tau}V(y(\tau),\tau,\sigma_\tau)d\tau\\ \nonumber
&&+\mathbb{E}\int_{0}^t e^{\zeta\tau}\mathcal LV(y(\tau),\tau,\sigma_\tau)d\tau\\\nonumber
&=& V(y_0,0,\sigma_0)+\mathbb{E} \int_{0}^t \zeta e^{\zeta\tau}y(\tau)^{\top}P(\sigma_\tau)y(\tau)d\tau\\\nonumber
&&+\mathbb E\int_{0}^t e^{\zeta\tau} y(\tau)^{\top}\{P(\sigma_\tau)[\gamma A(\sigma_\tau)-\beta I_n]\}^sy(\tau)d\tau\\\nonumber
&&+\frac{1}{2}\mathbb E\int_0^t e^{\zeta\tau}y(\tau)^{\top}\sum_{r=1}^{N'}q_{r,\sigma_\tau}P(r) y(\tau)d\tau\\\nonumber
&\leq&  V(y_0,0,\sigma_0).
\end{eqnarray*}
Hence, we have
\begin{eqnarray*}
\mathbb {E}\bigg[\|i(t)\|^2\bigg] &\leq& \mathbb{E}\bigg[\|y(t)\|^2\bigg]\leq \frac{2\mathbb{E} V(y(t),t,\sigma_t)}{a}\\
&\leq& \frac{2 V(y_0,0,\sigma_0)}{a}e^{-\zeta t},
\end{eqnarray*}
which implies that $\|i(t)\|$ converges to zero almost surely for all $v$. This completes the proof.
\end{proof}

\subsection{Characterizing Power of MTD without Considering Cost}

\begin{theorem}
\label{theorem:dynamic-structure-only}
Suppose configuration $\C_1=(G_1,\beta,\gamma)$ violates condition \eqref{eqn:static-epidemic-threshold} and MTD-induced configurations
$\C_l=(G_l,\beta,\gamma)$ for $l=2,\ldots,N'$ satisfy condition \eqref{eqn:static-epidemic-threshold}.
Denote by $\mu_l=\beta-\gamma \lambda_1(A_j)$ for $l=1,\ldots,N'$.
Without loss of generality, suppose $\mu_1<0<\mu_2<\cdots<\mu_{N'}$.
Under the definition of $Q$ in Theorem \ref{dsvsss-theorem},
the maximum portion of time the system can afford to stay in configuration $\C_1$ is
\begin{align}\label{dsvsss-pi-1}
\pi_1^*=\frac{\frac{b-1}{2b[-\mu_1+\delta]}}{\frac{b-1}{2b[-\mu_1+\delta]}+\frac{c-a}{2a[\mu_{N'}-\delta]}},
\end{align}
which is reached by launching MTD to induce configuration $(G_{N'}, \beta,\gamma)$.
That is, MTD is $(\mu_1,\cdots,\mu_{N'},\pi_1^*)$-powerful.
\end{theorem}

\begin{proof}
The infinitesimal generator $Q$ defined in the proof of
Theorem \ref{dsvsss-theorem} specifies the desired law $\sigma_t$, which can guide the deployment of MTD to force the overall dynamics converge to $I^*=(0,\ldots,0)$.
Note that $j$ in Eq. (\ref{dsvsss-theorem-order}) represents the number of configurations that violate condition \eqref{eqn:static-epidemic-threshold}.
Hence, $j=1$ in the present scenario.
For each $r\in\{1,\ldots,N'\}$, let $x_r=\frac{1}{-q_{rr}}$; then $\pi_r=\frac{x_r}{\sum_{p} x_p}$ is the portion of time in configuration $\C_r$.
Let $a,b,c$ be as defined in the proof of Theorem \ref{dsvsss-theorem}.

Consider configurations $\{\C_{1},\C_{k_{1}},\cdots,\C_{k_{m}}\}$, where $m\le N'$, and $\C_{k_l}=(G_{k_l},\beta,\gamma)$ for $l\in \{1,\cdots,m\}$
and $\{k_1,\ldots,k_m\}\subseteq\{2,\cdots,N'\}$ (which will be determined below) are MTD-induced configurations.
Under the definition of $Q$ in Theorem \ref{dsvsss-theorem}, we have
\begin{align*}
x_1\leq\frac{b-1}{2b[-\mu_1+\delta]}, ~~x_{k_l}\geq\frac{\frac{c+m-1}{m}-a}{2a[\mu_{k_l}-\delta]},~l=1,\cdots,m.
\end{align*}
This means that the dynamics converges to $I^*=(0,\ldots,0)$,
while staying in configuration $\mathcal C_1$ for a portion of time $\pi_1$, where
\begin{align*}
\pi_1&=\frac{x_1}{x_1+x_{k_1}+\cdots+x_{k_m}}\leq
\frac{x_1}{x_1+\sum_{l}\frac{\frac{c+m-1}{m}-a}{2a[\mu_{k_l}-\delta]}}\\
&\leq \frac{x_1}{x_1+\frac{c+m-1-am}{2a[\max_{l}\mu_{k_l}-\delta]}}\leq
\frac{x_1}{x_1+\frac{c-a}{2a[\mu_{k_{N'}}-\delta]}}\\
&\leq\frac{\frac{b-1}{2b[-\mu_1+\delta]}}{\frac{b-1}{2b[-\mu_1+\delta]}+\frac{c-a}{2a[\mu_{N'}-\delta]}}.
\end{align*}
We see that the maximum $\pi_1$, namely $\pi_1^*$, is reached when \\
$\{G_{k_1},\ldots,G_{k_m}\}=\{G_{N'}\}$. This completes the proof.
\end{proof}

Theorem \ref{theorem:dynamic-structure-only} further confirms that $\mu$ is indicative of the capability of a configuration in terms of
``forcing'' the overall dynamics to converge to $I^*=(0,\ldots,0)$.
Eq. (\ref{dsvsss-pi-1}) says that $\pi_1^*$ is monotonically increasing in $\mu_{N'}$ for fixed $\mu_1$ and decreasing in $\mu_1$ for fixed $\mu_{N'}$.
Figure \ref{fig:pi-1-dependence-2} confirms this property with
$a=0.8$, $b=1.5$, $c=2.4$, $\delta=10^{-5}$, while Eq. (\ref{dsvsss-pi-1}) leads to $\pi_1^*=\frac{\mu_{N'}-\delta}{\mu_{N'}-6\mu_1+5\delta}$.
\begin{figure}[H]
\centering
\includegraphics[width=0.46\textwidth]{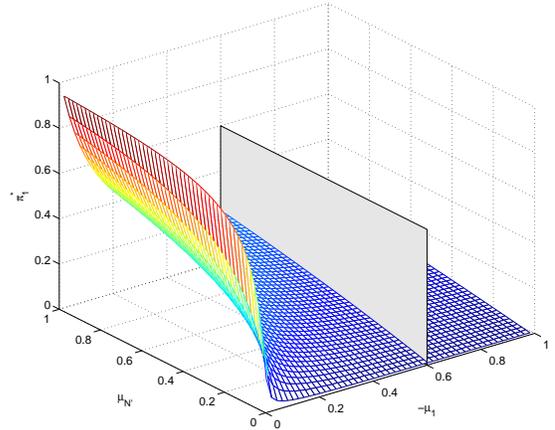}
\caption{Dependence of $\pi_1^*$ on $-\mu_1$ and $\mu_{N'}$.
}
\label{fig:pi-1-dependence-2}
\end{figure}

{\bf Remark}. In Theorem \ref{theorem:dynamic-structure-only} we consider a single configuration $\C_1$ that violates condition \eqref{eqn:static-epidemic-threshold}.
We can extend Theorem \ref{theorem:dynamic-structure-only} to accommodate multiple configurations that violate condition \eqref{eqn:static-epidemic-threshold},
because Theorem \ref{dsvsss-theorem} can accommodate this scenario.
However, if the goal is to maximize the time that the system can stay in the configurations that violate condition \eqref{eqn:static-epidemic-threshold},
the optimal solution is with respect to configuration $\C_j$, where $j$ is given by Eq. \eqref{dsvsss-theorem-order},
such that the system will stay in $\C_j$ for the maximum portion of time $\pi_1^*$ given by
Theorem \ref{theorem:dynamic-structure-only}.

\subsection{Characterizing Power of MTD while Considering Cost}
\label{ds_cost}

The goal is to make the dynamics converge to $I^*=(0,\ldots,0)$,
at minimum cost of launching MTD,
while the system stays a predetermined $\pi_1$ portion of time in $\C_1$.
The preceding case proved that $\pi_1 \leq \frac{\frac{b-1}{2b[-\mu_1+\delta]}}{\frac{b-1}{2b[-\mu_1+\delta]}+\frac{c-a}{2a[\mu_{N'}-\delta]}}$
is necessary to make the dynamics converge to $I^*=(0,\ldots,0)$.
Let $\mu_l=\beta-\gamma\lambda_1(A_l)$ for $l=1,\ldots,N'$.
Let the cost function $g(\cdot)$: ${\mathbb R}^+\to\mathbb R^+$ be the same as discussed in Definition \ref{definition-measure-with-cost},
namely that $g(\mu_l)$ is the cost of launching MTD to induce $\C_l$ for $l=2,\ldots,N'$,
where $g'(\mu)\geq 0$ for $\mu>0$.

Let $\sigma_t$ be the desired ``law" for deploying MTD, and denote by $Q=[q_{jk}]$ its infinitesimal generator.
Let $x_l=\frac{1}{-q_{ll}}$, where $\frac{1}{-q_{ll}}$ is the expectation of sojourn time in configuration $\C_l$.
Then, the portion of time in $\C_l$ is $\pi_{l}=\frac{x_{l}}{\sum_{l=1}^{N'}x_{l}}$. Our goal is to find the ``law" under which the cost is minimum.
Theorem \ref{dsvsss-theorem} specifies the desired Markov ``law" $\sigma_t$ via its infinitesimal generator.
Now we consider cost function $\Psi(x)$ with respect to $x=[x_{1},\cdots,x_{N'}]$.

Let $G_{k_l}$ for $l=1,\ldots,m'$ and $m'\geq 1$ be the MTD-induced configurations,
namely $\pi_k=0$ for $k\notin\{ 1,k_1,\cdots,k_{m'}\}$ and $\{k_1,\cdots,k_{m'}\}\subseteq\{2,\cdots,N'\}$.
Under the the definition of $Q$ in Theorem \ref{dsvsss-theorem}, $\sigma_t$ needs to satisfy:
\begin{align}
\label{dsvsss-opt-condition}
\left\{\begin{array}{lr}
0<x_1\leq \frac{b-1}{2b[-\mu_1+\delta]}\triangleq\bar{x}_1,& \\
x_{k_l}\geq\frac{\frac{c+m'-1}{m'}-a}{2a[\mu_{k_l}-\delta]}\triangleq\bar{x}_{k_l}(m'), ~~~l=1,\cdots, m'&\\
~~~~~~0<\delta\ll 1, a<1<b<c.&
\end{array}\right.
\end{align}
Note that $\frac{x_{1}}{x_1+\sum_{l=1}^{m'}x_{k_l}}$ and
$\frac{x_{k_l}}{x_1+\sum_{l=1}^{m'}x_{k_l}}$ are respectively the portions of time in configurations $\C_1$ and $\C_{k_l}$.
Define
\begin{eqnarray}
\bar{\pi}_{1}\triangleq\frac{\bar{x}_1}{\bar{x}_1+\sum_{l=1}^{m'}\bar{x}_{k_l}(m')},\label{ss1}
\end{eqnarray}
the maximum portion of time the system can stay in $\C_1$ when MTD induces the $m'$ configurations.
For any $\pi_1$ such that $\pi_{1}>\bar{\pi}_{1}$ does not hold,  $\pi_1$ cannot be realised by a underlying Markov process.
Therefore, we assume that $\pi_{1}^{*}<\bar{\pi}_{1}$.

Denote the index set corresponding to Eq. (\ref{ss1}) as 
\begin{align*}
\mathcal K=\bigg\{\{k_1,\cdots,k_{m'}\}|\pi_1^*\leq\frac{\bar{x}_1}{\bar{x}_1+\sum_{l=1}^{m'}
\bar{x}_{k_l}(m')}, k_1<\cdots< k_{m'}\bigg\}.
\end{align*}
For $\{k_1,\cdots,k_{m'}\}\in\mathcal K$, we need to find the ``law" $\sigma_t$
that satisfies (\ref{dsvsss-opt-condition}).
From $\pi_{1} =\frac{x_{1}}{x_1+\sum_{l=1}^{m'}x_{k_l}}$, we have
$\sum_{l=1}^{m'}x_{k_l}=\frac{1-\pi_1 }{\pi_1 }x_1$.
The cost of launching MTD according to ``law" $\sigma_t$ is:
\begin{eqnarray}
\label{dsvsss-opt-cost}
&&\Psi(x_1,x_{k_1},\cdots,x_{k_{m'}}) \nonumber\\
& =& \pi_1 g(\mu_1)+\sum_{l=1}^{m'}\pi_{k_l} g(\mu_{k_l}) \nonumber\\
&=&\pi_1 g(\mu_1)+(1-\pi_1^*)\frac{\sum_{l=1}^{{m'}}x_{k_l}g(\mu_{k_l})} {\sum_{l=1}^{{m'}}x_{k_l}}
\end{eqnarray}
subject to
\begin{align}
\label{dsvsss-opt-condition1}
\left\{\begin{array}{lr}
x_{k_l}\geq\frac{\frac{c+{m'}-1}{\ell}-a}{2a[\mu_{k_l}-\delta]}, ~~~l=1,\cdots,{m'}&\\
\sum_{l=1}^{{m'}}x_{k_l}=\frac{1-\pi_1 }{\pi_1 }x_1\leq\frac{1-\pi_1 }{\pi_1 }\bar{x}_1,&\\
~~~~~~~~~~~~0<\delta\ll 1, a<1<b<c.&
\end{array}\right.
\end{align}
We want to compute the minimize cost
\begin{align*}
\min_{x,\{k_1,\cdots,k_{m'}\}\in\mathcal K}\Psi(x_1,x_{k_1},\cdots,x_{k_{m'}}).
\end{align*}

\begin{theorem}
\label{dsvsss-theorem-B}
Given configuration $\C_1$ that violates condition \eqref{eqn:static-epidemic-threshold}
and MTD-induced configurations $\C_l$ for $l=2,\ldots,N'$ that satisfy condition \eqref{eqn:static-epidemic-threshold}.
Suppose $\pi_1 $, where $0<\pi_1 \leq \frac{\frac{b-1}{2b[-\mu_1+\delta]}}{\frac{b-1}{2b[-\mu_1+\delta]}
+\frac{c-a}{2a[\mu_{N'}-\delta]}},
$ is the portion of time that the system must stay in $\C_1$.
Denote by
\begin{eqnarray*}
&&G(k_1,\cdots,k_{m'})\\
&=&\frac{\sum_{l=1}^{\ell}\bar{x}_{k_l}(m')g(\mu_{k_l})+ g(\mu_{k_1})\Delta(k_1,\cdots,k_{m'})}{\sum_{l=1}^{m'}\bar{x}_{k_l}(m')+\Delta(k_1,\cdots,k_{m'})},\\
\end{eqnarray*}
where
\begin{eqnarray*}
\bar{x}_{k_{l}}(m')&=&\frac{\frac{c+m'-1}{\ell}-a}{2a[\mu_{k_l}-\delta]},\\
\Delta(k_1,\cdots,k_{m'})&=&\frac{1-\pi_1 }{\pi_1 }\bar{x}_1-\sum_{l=1}^{m'}\bar{x}_{k_l}(m').
\end{eqnarray*}
We want to find $\{k_1^*,\cdots,k_m^*\}$ such that
\begin{align}\label{dsvsss-withcost-opt}
\{\mu_{k_1^*},\cdots,\mu_{k_m^*}\}
=\arg\min_{\{k_1,\cdots,k_{m'}\}\in\mathcal K}G(k_1,\cdots,k_{m'})
\end{align}
For given cost function $g(\cdot)$ with arbitrary constant $g(\mu_1)$, the minimum cost is
\begin{eqnarray*}
&&\min_{x,\{k_1,\cdots,k_{m'}\}\in\mathcal K}\Psi(x_1,x_{k_1},\cdots,x_{k_{m'}}) \nonumber\\
&=&\Psi(\bar{x}_1,\bar{x}_{k_1^*}(m)+\Delta,\cdots,\bar{x}_{k_m^*}(m))\\
&=&\pi_1g(\mu_1)+(1-\pi_1 )G(k_1^*,\cdots,k_m^*),\nonumber
\end{eqnarray*}
which is reached by launching MTD to induce configuration \\
$\{(G_{k_l^*},\beta,\gamma)\}_{l=1}^{m}$ via the following deployment strategy:
\begin{align}\label{dsvsss-withcost-deployment}
&\pi_{k_1^*}=(1-\pi_1 )\frac{\bar{x}_{k_1^*}(m)+\Delta(k_1^*,\cdots,k_m^*)}
{\sum_{l=1}^{m}\bar{x}_{k_l^*}(m)+\Delta(k_1^*,\cdots,k_m^*)},\\
\nonumber&\pi_{k_l^*}=(1-\pi_1 )\frac{\bar{x}_{k_l^*}(m)}
{\sum_{l=1}^{m}\bar{x}_{k_l^*}(m)+\Delta(k_1^*,\cdots,k_m^*)}, l=2,\cdots,m.
\end{align}
Hence, MTD is $(\mu_1,\ldots,\mu_{N'},\pi_1 ,\Psi)$-powerful.
\end{theorem}

\begin{proof}
Suppose $\{k_1,\cdots,k_{m'}\}\in\mathcal K$.
We introduce variables $\zeta_{k_l}^2=x_{k_l}-\bar{x}_{k_l}(m')$ for $l=1,\cdots,m'$, $\zeta^2=\frac{1-\pi_1 }{\pi_1 }\bar{x}_1-\sum_{l=1}^{m'}x_{k_l}$
and translate the minimum problem specified by (\ref{dsvsss-opt-cost})-(\ref{dsvsss-opt-condition1}) into the following minimum problem:
\begin{align*}
\Psi(x_1,x_{k_1},\cdots,x_{k_{m'}})=\pi_1 g(\mu_1)+(1-\pi_1 )\frac{\sum_{l=1}^{m'}x_{k_l} g(\mu_{k_l})}{\sum_{l=1}^{m'}x _{k_l}}
\end{align*}
subject to
\begin{align*}
\left\{\begin{array}{lr}
x_{k_l}=\zeta_{k_l}^2+\bar{x}_{k_l}(m'), ~~~l=1,\cdots,m',&\\
\sum_{l=1}^{m'}x_{k_l}=\frac{1-\pi_1 }{\pi_1 }x_1,&\\
\sum_{l=1}^{m'}x_{k_l}+\zeta^2=\frac{1-\pi_1 }{\pi_1 }\bar{x}_1.&
\end{array}\right.
\end{align*}
Let $x=[x_{k_1},\cdots,x_{k_{m'}}]$, $\zeta=[\zeta_{k_1}\cdots,\zeta_{k_{m'}}]$, and \\
$\alpha=[\alpha',\alpha'',\alpha_{k_1}\cdots,\alpha_{k_{m'}}]$.
We study the Lagrange function
\begin{eqnarray*}\label{dsvsss-lagrange}
\Lambda_1(x,\zeta,\alpha) =\Psi(x_1,x_{k_1},\cdots,x_{k_{m'}})+ \\
\sum_{l=1}^{m'}\alpha_{k_l}[x_{k_l}-\zeta_{k_l}^2-\bar{x}_{k_l}(m')]+\alpha'\left(\sum_{l=1}^{m'}x_{k_l}-\frac{1-\pi_1 }{\pi_1 }x_1\right)+\\
\alpha''\left(\sum_{l=1}^{m'}x_{k_l}+\zeta^2-\frac{1-\pi_1 }{\pi_1 }\bar{x}_1\right).
\end{eqnarray*}
Find all stationary points $\{x,\zeta,\alpha\}$ of $\Lambda_1$, with gradient $\nabla\Lambda_1=0$. The $x$ parts of stationary points are
\begin{eqnarray*}
X_1&=&\left[\frac{\pi_1 \sum_{l=1}^{m'}\bar{x}_{k_l}(m')}{1-\pi_1 },\bar{x}_{k_1}(m'),\cdots,\bar{x}_{k_\ell}(m')\right],\\
X_l&=&\left[\bar{x}_1,\bar{x}_{k_1}(m'),\cdots,\bar{x}_{k_{m'}}(m')\right]+e_{l+1}\Delta(k_1,\cdots,k_{m'}),\\
&&~~~~~~~~~~~~~~ l=1,\cdots,m',
\end{eqnarray*}
where $e_l$ is the vector whose $l$-th element equals 1 and any other element equals 0.

By comparing the costs of these stationary points, we find the minimum cost of launching these $\ell$ configurations is
\begin{align*}
\pi_1 g(\mu_1)+(1-\pi_1 )G(k_1,\cdots,k_{m'}).
\end{align*}
This complete the proof.
\end{proof}

%% end ignore for proceedings version for saving space

{\bf Remark}. Similar to Theorem \ref{theorem:dynamic-structure-only},
in Theorem \ref{dsvsss-theorem-B}  we consider a single configuration $\C_1$ that violates condition \eqref{eqn:static-epidemic-threshold}.
We also can extend Theorem \ref{dsvsss-theorem-B} to accommodate multiple configurations that violate condition \eqref{eqn:static-epidemic-threshold},
because Theorem \ref{dsvsss-theorem} can accommodate this scenario.
The extension is straightforward because the portions of time that are allocated to the violating configurations are fixed and not involved in the definition of cost.

\subsection{Algorithm for Launching Optimal MTD}

Theorems \ref{theorem:dynamic-structure-only} and \ref{dsvsss-theorem-B} are constructive and lead to Algorithm \ref{alg:ds-ss-strategy} that can guide the deployment of optimal MTD.

\begin{algorithm}[hbt]
\caption{Launching optimal MTD (dynamic structures)\label{alg:ds-ss-strategy}}
INPUT: configuration $\C_1$, optional cost function $g(\cdot)$,
MTD-induced $\C_l$ for $l=2,\ldots,N'$ and $N'\geq 2$, constant $a>0$ (determining time resolution),
$\delta$ ($0< \delta <<1$), optional $\pi_1 $ \\
OUTPUT: Optimal MTD strategy

\begin{algorithmic}[1]
\IF{cost function is not given (i.e., no need to consider cost)}
\STATE Compute $\pi_1^*$ according to Eq. (\ref{dsvsss-pi-1})
\WHILE{TRUE}
\STATE Wait for time $T_1\leftarrow \EXP(a/\pi_1^*)$  \COMMENT{system in $\C_1$}
\STATE Launch MTD to force the system to stay in configuration $\C_{N'}$ for time $T_{N'}\leftarrow \EXP(a/(1-\pi_1^*))$
\STATE Stop launching MTD \COMMENT{system returns to $\C_1$}
\ENDWHILE
\ELSE 
\STATE Find indices $k_1^*,\cdots,k_m^*$ according to Eq. (\ref{dsvsss-withcost-opt})
\STATE Set $\pi_{k_1^*},\ldots, \pi_{k_m^*}$ as defined in Eq. (\ref{dsvsss-withcost-deployment}) and $\pi_k=0$ for $k\in \{2,\ldots,N'\}-\{k_1^*,\ldots,k_m^*\}$
\STATE Wait for time $T_1\leftarrow \EXP(a/\pi_1)$ \COMMENT{system in $\C_1$}
\STATE $k_j^*\leftarrow_R \{k_1^*,\cdots,k_m^*\}$
\WHILE{TRUE}
\STATE Launch MTD to make the system stay in $\C_{k_j^*}$ for time $T_{k_j^*}\leftarrow \EXP(a/\pi_{k_j^*})$
\STATE Set $\Delta= \{1,k_1^*,\cdots,k_m^*\} -\{k_j^*\} $
\STATE $k_j^* \leftarrow_R \Delta- \{k_j^*\}$ 
\IF{$k_j^*= 1$}
\STATE Stop launching MTD to make the system stay in $\C_1$ for time $T_1\leftarrow \EXP(a/\pi_{1})$
\STATE $k_j^*\leftarrow_R \{k_1^*,\cdots,k_m^*\}$
\ENDIF 
\ENDWHILE
\ENDIF
\end{algorithmic}
\end{algorithm}

In Algorithm \ref{alg:ds-ss-strategy}, lines 2-7 correspond to the case of not considering cost, where each step incurs $O(1)$ computational complexity.
Lines 9-21 correspond to the case of considering cost.
Specifically, line 9 incurs complexity $O(2^{N'})$, which is not infeasible because in practice $N'$ (i.e., the number of MTD-induced configurations)
is often small.
Possible instances of configurations the system will stay include
$\C_1,\C_{k_1^*},\ldots,\C_{k_m^*},\C_1^*,\ldots$ and \\
$\C_1,\C_{k_1^*},\C_{k_2^*},C_{k_1^*},\ldots,\C_{k_m^*},\C_1^*,\ldots$.

\section{Limitations of the Model}
\label{sec:extension-and-discussion}

First, the present study assumes that the attack-defense structures and parameters are given.
It is sufficient for characterizing the power of MTD.
Nevertheless, it is important to study how to obtain such structures and parameters.

Second, the present study does not allow the attacker to choose when to impose configuration $\C_1$.
It is important to give the attack the freedom in choosing when to impose $\C_1$.
This incurs technical difficulties. For example, the portion of time in the violating configuration may not be fixed at $\pi_{1}$, 
which breaks the setting of the optimisation problem. 

Third, it is interesting to extend the model to accommodate heterogeneous $\gamma_{v,u}$ and $\beta_v$.
However, this will make the model difficult to analyze mainly because of accommodating $\beta_v$.

\section{Conclusion}
\label{sec:conclusion}

We have introduced an approach of using cyber epidemic dynamics to
characterize the power of MTD. The approach offers algorithms for optimally deploying MTD,
where ``optimization" means maximizing the portion of time the system can afford to stay in an undesired configuration, 
or minimizing the cost of launching MTD when the system has to stay in an undesired configuration for a predetermined portion of time.
We have discussed the limitations of the present study, which should inspire fruitful future research.

\medskip

\noindent{\bf Acknowledgement}.
Wenlian Lu was supported by a Marie Curie International Incoming Fellowship from
European Commission (no. FP7-PEOPLE-2011-IIF-302421), National Natural Sciences
Foundation of China (no. 61273309), and Program for New Century Excellent Talents in University (no. NCET-13-0139).
Shouhuai Xu was supported in part by ARO Grant \#W911NF-12-1-0286 and AFOSR Grant FA9550-09-1-0165.
Any opinions, findings, and conclusions or recommendations expressed in this material are those of
the author(s) and do not necessarily reflect the views of any of the funding agencies.

\appendix

Now we present the proof of Theorem  \ref{theorem:sd-ss-minimum-cost}.

\begin{proof}
For any $m\geq k^*$, we have $\pi_1 \mu_1+(1-\pi_1 )\mu_m>0$, meaning that the dynamics will converge to $I^*=0$ by launching MTD to
induce configuration $\C_{m}$ with portion of time $1-\pi_{1}^{*}$.
For any $l<k^*$, we have $\pi_1 \mu_1+(1-\pi_1 )\mu_l\leq0$, meaning that condition (\ref{corollary:ave1}) for the dynamics to converge to $I^*=0$ is
not satisfied even if the defender launches MTD to induce configuration $\C_l$ with portion of time $1-\pi_{1}^{*}$.

If $k^*=2$, the dynamics will converge to $I^*=0$ by launching MTD to induce configuration $\C_{l}, l\geq 2$. Since $f(\cdot)$ is non-decreasing, we have
\begin{eqnarray*}
\Phi(\pi_{2},\cdots,\pi_{N})&=&\pi_1 f(\mu_1)+\sum_{l=2}^N\pi_l f(\mu_l)\\
&\geq&\pi_1 f(\mu_1)+(1-\pi_1 )f(\mu_2).
\end{eqnarray*}
It can be seen that the equality above can be guaranteed by taking $\pi_{2}=1-\pi_{1}^{*}$ and $\pi_{j}=0$ for $j>2$.

If $k^*>2$, we use Lagrange multipliers to calculate the minimum cost $\Phi(\pi)$, subject to constraint (\ref{restrict_con3}). Since Lagrange multipliers require equality constraints,
we introduce vector of variables $x=[x_2,\cdots,x_N]$ with $x_j$ satisfying $\pi_j=x^2_j$ and variable $\zeta$
with $\zeta^2=\pi_1 \mu_1+\sum_{j=2}^{N}\mu_j x_j^2-\delta$, such that solving the minimization problem
is equivalent to finding the minimum of the following function
$\Phi_1(x)=\pi_1 f(\mu_1)+\sum_{j=2}^N x_j^2 f(\mu_j)$
subject to
\begin{eqnarray*}
 h_1(x)&=&\sum_{j=2}^{N}\mu_j x_j^2=-\pi_1 \mu_1+\delta+\zeta^2\\
h_2(x)&=&\sum_{j=2}^{N}x_j^2= 1-\pi_1 ,
\end{eqnarray*}
where $\Phi_1$,$h_1$, and $h_2$ have continuous first partial derivatives.
To solve this variant problem, we introduce Lagrange multipliers $\alpha_1$ and $\alpha_2$ and the Lagrange function as follows
\begin{align*}
\Lambda(x,\zeta,\alpha_1,\alpha_2)=&\Phi_1(x)+\alpha_1[h_1(x)+\pi_1 \mu_1-\delta-\zeta^2]\\
&+\alpha_2[h_2(x)+\pi_1 -1].
\end{align*}
Denote by $\Phi_1(x_0)$ the minimum of $\Phi_1(x)$.
There exist $\bar{\alpha}_1$ and $\bar{\alpha}_2$ such that $(x_0,\zeta_0,\bar{\alpha}_1,\bar{\alpha}_2)$ is a stationary point for the Lagrange function $\Lambda(x,\zeta,\alpha_1,\alpha_2)$,
i.e., with gradient $\nabla\Lambda=0$. Then, we are to solve $\nabla\Lambda=0$:
\begin{eqnarray}\label{gradient}
\left\{\begin{array}{lr}
\nabla_{x_2}\Lambda=2[f(\mu_2)+\alpha_1\mu_2+\alpha_2]x_2=0\\
~~~~~~~~~\vdots\\
\nabla_{x_N}\Lambda=2[f(\mu_N)+\alpha_1\mu_N+\alpha_2]x_N=0\\
\nabla_{\zeta}\Lambda=-2\alpha_1\zeta=0\\
\nabla_{\alpha_1}\Lambda=\sum_{j=2}^N\mu_j x_j^2+\pi_1 \mu_1-\delta-\zeta^2=0\\
\nabla_{\alpha_2}\Lambda=\sum_{j=2}^N x_j^2+\pi_1 -1=0,
\end{array}\right.
\end{eqnarray}

To solve (\ref{gradient}), there are three cases:

{\em Case 1: The optimal strategy is that the defender launches MTD to induce only one configuration $\C_m$.}
Then, $m\geq k^*$ must hold and the cost is $\Phi_1(x)=\pi_1 f(\mu_1)+(1-\pi_1 )f(\mu_m)$. Hence, the minimum cost of launching MTD to induce a single configuration is
\begin{equation}\label{sdss-cost1}
\pi_1 f(\mu_1)+(1-\pi_1 )f(\mu_{k^*}),
\end{equation}
which is reached when inducing configuration $\C_{k^*}$.

{\em Case 2: The optimal strategy is that the defender launches MTD to induce
two configurations $\C_l,\C_m$.} The minimum cost will be reached at $l<k^*\leq m$ because the cost function $f(\mu)$
is non-decreasing in $\mu>0$. Firstly, we look for the minimal cost when launching MTD to induce configurations $\C_l,\C_m, l< k^*\leq m$.
This requires to find stationary points of $\Lambda(x,\zeta,\alpha_1,\alpha_2)$ such that $x_l^2>0,x_m^2>0$ and
${x}_k^2=0$ for $k\geq 2, k\neq l,m$.
It can be verified that the following points are stationary points for $\Lambda(x,\zeta,\alpha_1,\alpha_2)$:
\begin{align*}
\left[\begin{array}{c}
x_l^2\\
x_m^2
\end{array}\right]=
\frac{1}{\mu_k-\mu_1}
&\left[\begin{array}{c}
(\mu_{m}-\delta)+\pi_1 (\mu_1-\mu_{m})\\
-(\mu_{l}-\delta)+\pi_1 (\mu_{l}-\mu_1)
\end{array}\right],\\
&~\zeta=0,~ {x}_k^2=0,~k\geq 2,~ k\neq l,m.
\end{align*}
and the corresponding cost is $F(\mu_l,\mu_m)$. Hence, the minimum cost of inducing two configurations is
\begin{align}\label{sdss-cost2}
\min_{l< k^*\leq m}F(\mu_l,\mu_m),
\end{align}
which is reached by launching MTD to induce configurations $\C_{l^*},\C_{m^*}$ according to $\pi_{l^*},\pi_{m^*}$ in Eq. (\ref{eqn:sdvsss-deploy}).

{\em Case 3: The optimal strategy is that the defender launches MTD to induce $m'\geq3$ configurations $\C_{k_j}, 1\leq j\leq m'$.}
To find the minimum cost, we need to find stationary points of $\Lambda(x,\zeta,\alpha_1,\alpha_2)$
such that $x_{k_j}^2>0, 1\leq j\leq m'$ and ${x}_k^2=0$ for $k\geq 2, k\neq k_j$. That is,
\begin{align}\label{multiple_points}
\left\{\begin{array}{lr}
f(\mu_{k_1})+{\alpha}_1\mu_{k_1}+{\alpha}_2=0\\
~~~~~~~~~\vdots\\
f(\mu_{k_{m'}})+{\alpha}_1\mu_{k_{m'}}+{\alpha}_2=0.\\
\end{array}\right.
\end{align}
Thus, the stationary point $x$ should satisfy
\begin{align*}
\left\{\begin{array}{lr}
x_{k_1}^2\mu_{k_1}+\cdots+\bar{x}_{k_{m'}}^2\mu_{k_{m'}}=\delta-\pi_1 \mu_1\\
\bar{x}_{k_1}^2+\cdots+\bar{x}_{k_{m'}}^2=1-\pi_1 \\
\bar{x}_{k_{l}}^2>0,~1\leq l\leq m',~{x}_k^2=0,~k\geq 2, k\neq k_l\\
\end{array}\right.,~ \zeta=0.
\end{align*}
The cost at this stationary points $x$ becomes
\begin{align}\label{sdss-cost3}
 \nonumber\Phi=&x_{k_1}^2f(\mu_{k_1})
+\cdots+x_{k_{m'}}^2f(\mu_{k_{m'}})\\
\nonumber=&-(\delta-\pi_1 \mu_1)\alpha_1-(1-\pi_1 )\alpha_2\\
=&F(\mu_{k_1},\mu_{k_2}).
\end{align}
If (\ref{multiple_points}) does not hold, then there is no stationary point $x$ in the form $x_{k_j}^2>0, 1\leq j\leq m'$ and
${x}_k^2=0$ for $k\geq 2, k\neq k_j$, meaning that there is no minimum cost when inducing these configurations.

By comparing the costs given by Eqs. (\ref{sdss-cost1}), (\ref{sdss-cost2}), and (\ref{sdss-cost3}), we conclude that
the minimum cost is given by Eq. (\ref{sdss-cost2}). This completes the proof.
\end{proof}

\end{document}